\documentclass[a4paper]{rsproca}

\usepackage{amsmath, amssymb, amsfonts, amsthm}
\usepackage{graphicx}
\usepackage{endfloat, endnotes, setspace, verbatim, geometry}
\usepackage{times, helvet, courier}
\usepackage{bm}
\usepackage{url} 
\usepackage[english]{babel}
\usepackage{dcolumn}

\newtheorem{theorem}{\bf Theorem}[section]

\newtheorem{remark}{\bf Remark}[section]
\newtheorem{proposition}{\bf Proposition}[section]

\newcommand{\tab}{\hspace{5mm}}
\newcommand{\mbc}{\mathbb{C}}
\newcommand{\mbh}{\mathbb{H}}
\newcommand{\mbr}{\mathbb{R}}
\newcommand{\mfa}{\mathfrak{a}}
\newcommand{\mfb}{\mathfrak{b}}
\newcommand{\mcd}{\mathcal{D}}
\newcommand{\mcl}{\mathcal{L}}
\newcommand{\mcm}{\mathcal{M}}
\newcommand{\Res}{\text{Re}(s)}

\jname{rspa} 
\Journal{Proc R Soc A\ }

\begin{document}
\title{Towards Quantized Number Theory: Spectral Operators and an Asymmetric Criterion for the Riemann Hypothesis}
\author{Michel L. Lapidus}
\address{University of California\\
Department of Mathematics\\
900 University Ave.\\
Riverside, CA 92521-0135, USA}

\subject{Mathematics, mathematical and theoretical physics, fractal geometry, spectral geometry, number theory, operator theory, functional analysis, ordinary and partial differential equations, spectral theory, quantum mechanics, quantum field theory, statistical physics, noncommutative geometry.}

\keywords{Riemann zeta function, Riemann hypothesis (RH), quantization, quantized number theory, fractal strings, geometry and spectra, Minkowski dimension, Minkowski measurability, complex dimensions, Weyl--Berry conjecture, fractal drums, infinitesimal shift, spectral operator, invertbility, quantized Dirichlet series and Euler product, universality, phase transitions, symmetric and asymmetric criteria for RH.}

\corres{Michel L. Lapidus \\
\email{lapidus@math.ucr.edu} }

\begin{abstract}
This research expository article contains a survey of earlier work (in \S2--\S4) but also contains a main new result (in \S5), which we first describe. Given $c \geq 0$, the spectral operator $\mfa = \mfa_c$ can be thought of intuitively as the operator which sends the geometry onto the spectrum of a fractal string of dimension not exceeding $c$. Rigorously, it turns out to coincide with a suitable quantization of the Riemann zeta function $\zeta = \zeta(s)$: $\mfa = \zeta (\partial)$, where $\partial = \partial_c$ is the infinitesimal shift of the real line acting on the weighted Hilbert space $L^2 (\mathbb{R}, e^{-2ct} dt)$. In this paper, we establish a new asymmetric criterion for the Riemann hypothesis, expressed in terms of the invertibility of the spectral operator for all values of the dimension parameter $c \in (0, 1/2)$ (i.e., for all $c$ in the left half of the critical interval $(0,1)$). This corresponds (conditionally) to a mathematical (and perhaps also, physical) ``phase transition'' occurring in the midfractal case when $c= 1/2$. Both the universality and the non-universality of $\zeta = \zeta (s)$ in the right (resp., left) critical strip $\{1/2 < \Res < 1 \}$ (resp., $\{0 < \Res < 1/2 \}$) play a key role in this context. These new results are presented in \S5. In \S2, we briefly discuss earlier joint work on the complex dimensions of fractal strings, while in \S3 and \S4, we survey earlier related work of the author with H. Maier and with H. Herichi, respectively, in which were established symmetric criteria for the Riemann hypothesis, expressed respectively in terms of a family of natural inverse spectral problems for fractal strings of Minkowski dimension $D \in (0,1),$ with $D \neq 1/2$, and of the quasi-invertibility of the family of spectral operators $\mfa_c$ (with $c \in (0,1), c \neq 1/2$).
\end{abstract}

%

\maketitle

\tableofcontents

\section{Introduction}
Our main goal in this paper is to first briefly explain (in \S2) the general ideas and results concerning fractal strings and their complex dimensions which led to the intimate connections (discussed in \S3) between the vibrations of fractal strings and the Riemann zeta function $\zeta = \zeta (s)$ and then, the Riemann hypothesis, discovered in the early 1990s by the author and his collaborators (Carl Pomerance [LapPom1,2] and Helmut Maier [LapMai1,2], respectively). The resulting geometric interpretation of the critical strip $0 \leq \Res \leq 1$ for $\zeta = \zeta(s)$ was made entirely rigorous by means of the mathematical theory of complex dimensions developed by the author and Machiel van Frankenhuijsen in a series of research monographs [Lap-vFr1--3] (and papers). Due to space limitations, we will not be able to do justice to this theory, for which the interested reader can refer to the latest book in the above series, [Lap-vFr3], and many relevant references therein. \\

\tab In \S 6.3.1 and \S 6.3.2 of [Lap-vFr2,3] was introduced, at the semi-heuristic level, the notion of spectral operator $\mfa = \mfa_c$, which sends the geometry onto the spectrum of a fractal string. Intuitively, the parameter $c$ plays here the role of the upper bound  for the Minkowski dimensions of the fractal strings on which $\mfa = \mfa_c$ acts. In the forthcoming book, \cite{HerLap1}, and in an associated series of papers, [HerLap2--5], Hafedh Herichi and the author have developed a mathematical theory of the spectral operator, within a rigorous functional analytic framework. In particular, they have precisely determined its spectrum as well as that of the closely related infinitesimal shift (of the real line) $\partial = \partial_c$, viewed as the first order differential operator $d/dt$ acting on a suitable Hilbert space of functions, the weighted $L^2$-space $\mathbb{H}_c = L^2 (\mathbb{R}, e^{-2ct}dt)$, where $c \in \mathbb{R}$ is arbitrary. It turns out that $\mfa_c = \zeta(\partial_c)$ or, more briefly, $\mfa = \zeta(\partial),$ in the sense of the functional calculus for the unbounded normal operator $\partial = \partial_c.$ In this manner, the spectral operator $\mfa = \mfa_c$ can be viewed as a suitable quantization of the Riemann zeta function. Actually, in \cite{HerLap1} (see also \cite{HerLap5}), a quantized (i.e., an operator-valued) Dirichlet series, Euler product representation and ``analytic continuation'' of $\mfa_c = \zeta (\partial_c)$ are obtained (for suitable values of $c$ and/or under appropriate assumptions). \\

\tab In this context, the question of the invertibility of the spectral operator is a key question. It is shown in \cite{HerLap1} (see also [HerLap2,3]) that $\mfa = \mfa_c$ is ``quasi-invertible'' for all $c \in (0,1), c \neq 1/2$ [equivalently, for all $c \in (0, 1/2)$ or for all $c \in (1/2, 1),$ respectively] if and only the Riemann hypothesis is true (i.e., iff $\zeta (s) = 0$ with $0 < \Res < 1$ implies that $\Res = 1/2$). This result provides a natural operator theoretic counterpart of the result of the author and H. Maier in [LapMa1,2] showing that a suitable inverse spectral problem, (ISP)$_D$, for fractal strings (of dimension $D$) is true in all possible dimensions $D$ in $(0,1), D \neq 1/2$ (or, equivalently, due to the functional equation of $\zeta$, for all $D \in (0, 1/2)$ or else for all $D \in (1/2, 1)$, respectively) if and only if the Riemann hypothesis (RH) is true.\\

\tab In \cite{HerLap1} (see also [HerLap2,3]), the spectrum of $\mfa = \mfa_c$ is shown to obey several ``mathematical phase transitions''; one conditionally (i.e., if RH holds), in the midfractal case (in the sense of [Lap1--3]) when $D = 1/2$, and another one unconditionally (i.e., independently of the truth of RH), at $D=1$. The first phase transition is intimately connected with the universality of the Riemann zeta function (see [HerLap1,4]). A brief exposition of some of the most relevant results of [HerLap1--5] is provided in \S4 of the present article.\\

\tab The spectral reformulation of the Riemann hypothesis obtained in \cite{LapMai2}, revisited and extended in [Lap-vFr1--3] in the light of the mathematical theory of complex dimensions, and given a rigorous operator theoretic version in [HerLap1--3], is a symmetric criterion for RH. Indeed, due to the functional equation satisfied by $\zeta = \zeta(s)$ (and thus connecting $\zeta(s)$ and $\zeta (1-s)$), it can be equivalently formulated for all values of the underlying parameter ($D$ in [LapMai1,2] and [Lap-vFr1--3] or $c$ in [HerLap1--5]). By contrast, the new criterion for the Riemann hypothesis obtained in this paper is asymmetric, in the sense that  it is only stated (and valid) for all values of the underlying parameter $c$ in the open interval $(0, 1/2)$. In fact, its counterpart for all values of $c$ in the symmetric interval $(1/2, 1)$ cannot be true, due to the universality of the Riemann zeta function $\zeta = \zeta(s)$ in the right critical strip $1/2 < \Res < 1$. More specifically, the main new result of this paper is the following (see \S5 and Theorem \ref{T1} below): The spectral operator $\mfa = \mfa_c$ is invertible (in the usual sense of the invertibility of an unbounded operator) for all values of $c$ in $(0, 1/2)$ if and only if the Riemann hypothesis is true. Again, this new {\em asymmetric} spectral reformulation of RH does not have any analog for the interval $(1/2, 1)$ because the spectral operator is not invertible for any value of $c$ in $(1/2, 1)$. This raises the question of finding an appropriate physical interpretation for this asymmetry, which amounts to finding ``the'' physical origin of the phase transition (conditionally) occurring in the midfractal case where $c = 1/2$. \\

\tab In the previous context of the results of [LapMai1,2], early speculations were made by the author in [Lap2,3] to find a physical interpretation of the corresponding mathematical phase transition at $D=1/2$ in terms of Wilson's notion of complex dimension within the realm of his theory of phase transitions and critical phenomena in quantum statistical physics and quantum field theory \cite{Wil}. We leave it to the interested reader to find his or her physical interpretation of the asymmetric criterion for RH obtained in the present paper. Our own expectation is that physically, this asymmetry or phase transition finds its source in the symmetry breaking associated to an appropriate (and possibly yet to be discovered) quantum field theory. (Unlike in [BosCon1,2] or [Con, \S V.11], it would occur at $1/2$ rather than at $1$, the pole of $\zeta$.) Presumably, this quantum field theory should be intimately associated with number theory (or ``arithmetic''), geometry, dynamics and spectral theory, and probably formulated in the spirit of the conjectural picture for the generalized Riemann hypothesis (GRH) proposed in the author's earlier book, {\em In Search of the Riemann Zeros}: {\em Strings, Fractal Membranes and Noncommutative Spacetimes} \cite{Lap6}, in terms of a Riemann flow on the moduli space of quantized fractal strings (called ``fractal membranes'') and the associated flows of zeta functions (or partition functions) and of their zeros. However, we invite the readers to follow their own intuition and to be led into unknown mathematical and physical territory, wherever their own imagination will guide them.\\

\tab In closing this introduction, we stress that the main goal of this paper (beside providing in \S2--\S4 a brief survey of earlier relevant work which partly motivates it) is to obtain the following theorem, in \S5. (See Theorem \ref{T5.4} and the comments surrounding it for additional information.) Let $\mfb = \mfa \mfa^* = \mfa^* \mfa$ be the nonnegative self-adjoint operator naturally associated with the normal operator $\mfa$. For notational simplicity, we write $\mfa$ and $\mfb$ instead of $\mfa_c$ and $\mfb_c$, respectively. Note that for $c \leq 1$, both $\mfa$ and $\mfb$ are unbounded linear operators (according to the results of [HerLap1--4] discussed in \S4).

\begin{theorem}[Asymmetric criterion for RH]\label{T1}
Each of the following statements is equivalent to the Riemann hypothesis $($RH$):$
\begin{enumerate}
\item[$($i$)$] For every $c \in (0,1/2)$, the spectral operator $\mfa$ is invertible. 
\item[$($ii$)$] For every $c \in (0,1/2)$, the unbounded operator $\mfb$ is invertible.
\item[$($iii$)$] For every $c \in (0, 1/2)$, the nonnegative self-adjoint operator $\mfb$ is bounded away from zero.
\end{enumerate}
\end{theorem}

\tab Finally, we point out that for the reader's convenience, we have included (just before the references) a glossary of some of the main notation and terminology used in this paper.

\section{Geometry and Spectra of Fractal Strings}

In this section, we give a brief overview of the theory of fractal strings and their complex dimensions (see \cite{Lap-vFr3}). 

\subsection{Geometric zeta function and Minkowski dimension.}

A {\em fractal string} is a one-dimensional drum with fractal boundary. Hence, it can be realized geometrically  as a bounded open subset of the real line: $\Omega \subset \mathbb{R}$. As is well known, $\Omega$ is then the disjoint union of its connected components (open intervals). Therefore, we can write it as a countable union of disjoint (bounded) open intervals $I_j, j=1,2, \cdots : \ \Omega = \bigcup_{j=1}^\infty I_j.$ Denoting by $\ell_j$ the length of $I_j$ ($\ell_j = |I_j|$), we may assume without loss of generality that $\mathcal{L} = \{\ell_j\}_{j=1}^\infty$ is nonincreasing and (if $\mathcal{L}$ is infinite, which will always be assumed from now on, in order to avoid trivial special cases) that $\ell_j \downarrow 0 \text{ as } j \rightarrow \infty$. Since all of the key notions considered in fractal string theory are independent of the geometric realization $\Omega$, we may consider that a fractal string is equivalently given by the sequence of lengths (or {\em scales}) $\mathcal{L} = \{\ell_j \}_{j=1}^\infty$, with the distinct lengths written in nonincreasing order and according to their multiplicities. \\

\tab The geometric point of view, however, is essential to motivate and define the spectral problem associated with $\mcl$, as well as to connect and contrast the theory of fractal strings with the higher-dimensional geometric and spectral theory of fractal drums.\\

\tab To a given fractal string $\mathcal{L} = \{\ell_j \}_{j=1}^\infty$ (or $\Omega \subset \mathbb{R}$), we associate the {\em geometric zeta function} of $\mathcal{L}$, denoted by $\zeta_{\mathcal{L}}$, and defined (for $s \in \mathbb{C}$ with $\Res$ large enough) by $\zeta_{\mathcal{L}} (s) = \sum_{j=1}^\infty \ell_j^s.$ As it turns out, $\zeta_{\mathcal{L}}$ is holomorphic in the open, right half-plane  $\Pi := \{\Res > D \} $ (the set of all $s \in \mathbb{C}$ such that $\Res > D$), and $\Pi$ is the largest such half-plane. Here, the unique extended real number $D \in \mathbb{R} \cup \{\pm \infty \}$ is called the {\em abscissa of convergence} of $\zeta_\mathcal{L}$: $D = \inf \{\alpha \in \mathbb{R}: \sum_{j=1}^\infty \ell_j^\alpha < \infty \}$; see, e.g., [HardWr, Ser]. Early on (see [Lap2--3] and [Lap-vFr3, \S 1.2]), the author observed that $D$ coincides with the Minkowski dimension of $\mathcal{L}$ (i.e., of the boundary $\partial \Omega$ of any of its geometric realizations $\Omega$):  $D = D_M$. Recall that this notion of fractal dimension (often also called ``box dimension'' in the applied literature) and very different mathematically and physically from the classic Hausdorff dimension (see \cite{Lap1}), can be defined as follows: Given $\varepsilon > 0$, let $\Omega_\varepsilon = \{x \in \Omega: \text{dist} (x, \partial \Omega) < \varepsilon \}$ denote the $\varepsilon$-neighborhood (or {\em tubular neighborhood}) of $\Omega$, and let $V (\varepsilon) = |\Omega_{\varepsilon}|$, the volume (i.e., length or $1$-dimensional Lebesgue measure) of $\Omega_\varepsilon$. (It is shown in [LapPom1,2] that for a fixed $\varepsilon > 0$, $V(\varepsilon)$ depends only on $\mathcal{L} = \{\ell_j\}_{j=1}^\infty$. Hence, as was mentioned earlier, all of the notions discussed here depend only on $\mathcal{L}$ and not on the geometric representation of $\mathcal{L}$ by $\Omega \subset \mathbb{R}$.) Then, given $d \geq 0$, let $\mathcal{M}_d^*$ (resp., $\mathcal{M}_{*,d}$) denote the upper (resp., lower) limit of $V(\varepsilon)/\varepsilon^{1-d}$ as $\varepsilon \rightarrow 0^+$, called, respectively, the $d$-{\em dimensional upper} (resp., {\em lower}) {\em Minkowski content} of $\mcl$. Clearly, we always have $0 \leq \mathcal{M}_{*,d} \leq \mathcal{M}_d^* \leq \infty$. \\

\tab The {\em Minkowski dimension} of $\mathcal{L}$ (or, equivalently, of its boundary $\partial \Omega$), denoted by $D_M = D_M (\mathcal{L})$, is then given by $D_M := \inf \{d \geq 0 : \mathcal{M}_d^* < \infty\} = \sup \{d \geq 0: \mathcal{M}_d^* = + \infty\}.$ In fact, $D_M$ is the unique real number such that $\mathcal{M}_d^* = + \infty$ for $d < D_M$ and $\mathcal{M}_d^* = 0$ for $d > D_M$. (See, e.g., [Fal1, Fed, Mat, Lap-vFr3, Tri].) Physically, it can be thought of as a critical exponent since it gives information about the behavior of $V(\varepsilon)$ as a power of $\varepsilon$. In the sequel, for brevity and in light of the aforementioned identity $D = D_M$, we will use $D$ instead of $D_M$ to denote the Minkowski dimension of $\mathcal{L}$. A fractal string $\mathcal{L}$ (or, equivalently, its boundary $\partial \Omega$) is said to be {\em Minkowski nondegenerate} if $0 < \mathcal{M}_* (\leq) \mathcal{M}^* < \infty$, where $\mathcal{M}^* := \mathcal{M}_D^*$ and $\mathcal{M}_* := \mathcal{M}_{*,D}$ denote, respectively, the {\em upper} and the {\em lower Minkowski content} of $\mathcal{L}$ (i.e., of $\partial \Omega$). In addition, if $\mathcal{M_* = \mathcal{M}^*}$, we denote by $\mathcal{M}$ this common value and say that $\mathcal{L}$ (i.e., $\partial \Omega$) is {\em Minkowski measurable} and has {\em Minkowski content} $\mathcal{M}$. Hence, $\mathcal{L}$ is Minkowski measurable with Minkowski content $\mathcal{M}$ if and only if the following limit exists in $(0, + \infty)$ and $\lim_{\varepsilon \rightarrow 0^+} V (\varepsilon)/\varepsilon^{1-D} = \mathcal{M}.$

\subsection{Geometric interpretation of the critical strip.}

For a fractal string, we always have $0 \leq D = D_M \leq 1$; that is, $D = D_M$ lies in the (closed) {\em critical interval} $[0,1]$, where (in the terminology of \cite{Lap1}, later adopted in [Lap2--3, 6--9], [LapPom1--3], [LapMai1--2], \cite{HeLap}, [Lap-vFr1--3] and [HerLap1--5]), the case when $D=0$ (resp., $D=1$) is referred to as the {\em least} (resp., {\em most}) {\em fractal case}, while the case when $D = 1/2$ is called the {\em midfractal case.} As we shall see, the latter midfractal case $D=1/2$ plays a key role in the theory, in connection with the Riemann zeta function and the Riemann hypothesis. Recall that according to the celebrated Riemann hypothesis \cite{Rie} (translated in [Edw, App.]), the {\em critical zeros} (also called the {\em nontrivial zeros}) of the Riemann zeta function $\zeta = \zeta(s)$, that is, the zeros of $\zeta$ located in the {\em critical strip} $0 < \Res < 1$, all lie on the {\em critical line} $\{\Res = 1/2 \}$. Equivalently, $\zeta(s) = 0$ with $0 \leq \Res \leq 1$ implies that $\Res = 1/2.$ \\

\tab Indeed, towards the very beginning of the mathematical theory of fractal strings (starting with [Lap2--3], [LapPom1--3] and [LapMai1--2]), one conjectured the existence of a notion of {\em complex dimension} such that the left-hand side of the (closed) critical strip $0 \leq \Res \leq 1$ (the vertical line $\{\Res = 0 \}$) corresponds to the least fractal (or nonfractal) case $D=0$, the right-hand side (the vertical line $\{\Res = 1 \}$) corresponds to the most fractal case $D=1$, while the middle of the critical strip (the {\em critical line} $\{ \Res = 1/2 \}$) corresponds to the midfractal case $D = 1/2$. The answer turned out to be positive, as was intuitively clear after the works in [LapPom1--2], [Lap2--3] and, especially, in [LapMai1--2], and as was fully and rigorously justified in the mathematical theory of complex dimensions developed in [Lap-vFr1--3] (as well as in earlier and later papers by the same authors). We refer to the research monograph [Lap-vFr3] for a complete exposition of the theory of complex dimensions of fractal strings. We will briefly introduce it in \S 2$(d)$ below, but first of all, we wish to explain the origins of the connections between the Riemann zeta function and fractal strings, viewed as vibrating objects.

\subsection{From the geometry to the spectrum of a fractal string, via $\zeta = \zeta(s).$}

A fractal string, $\Omega = \cup_{j=1}^\infty I_j$, can be viewed as countably many ordinary strings (the intervals $I_j$ of lengths $\ell_j$), vibrating independently of one another. Equivalently, the strings $I_j$ can be thought of as the strings of a fractal harp. The (normalized frequency) spectrum $\sigma (\mathcal{L})$ of the fractal string $\mathcal{L}$ is given by the union (counting multiplicities) of the spectra of the intervals $I_j: \sigma (\mathcal{L}) = \bigsqcup_{j=1}^\infty \sigma (I_j)$. Note that the fundamental frequency of the $j$-th string $I_j$ is equal to $\ell_j^{-1}$; so that for each fixed integer $j \geq 1$, we have $\sigma (I_j) = \{n \cdot \ell_j^{-1} : n \geq 1 \}$ and hence, $\sigma (\mathcal{L}) = \{n \cdot \ell_j^{-1} : n \geq 1, j \geq 1 \}.$ Therefore, the {\em spectral zeta function} of $\mathcal{L}$ (or, equivalently, of the Dirichlet Laplacian on $\Omega$), denoted by $\zeta_{\nu}$ and defined by $\zeta_{\nu} (s) = \sum_{f \in \sigma (\mathcal{L})} f^{-s},$ for $\Res$ large enough, is given by the following product formula:
\begin{align}\label{2.8}
\zeta_{\nu} (s) = \zeta (s) \cdot \zeta_{\mathcal{L}}(s),
\end{align}
where $\zeta_{\mathcal{L}}$ is the geometric zeta function of $\mathcal{L}$ (given by its expression in \S 2$(a)$ or by its meromorphic continuation thereof). Formula \eqref{2.8} was first observed in [Lap2--3]. It, along with its counterpart just below, was key to our later developments. (We note that formula \eqref{2.8} was later reinterpreted and extended to various settings, in terms of complex dynamics; see [Tep] and [LalLap1--2].) At the level of the counting functions, \eqref{2.8} becomes $N_{\nu} (x) = \sum_{n=1}^{\infty} N_{\mathcal{L}} (x/n),$ for $x > 0, $ where $N_\mathcal{L}$ (resp., $N_{\nu}$) denotes the {\em geometric} (resp., {\em spectral} or {\em frequency}) {\em counting function} of $\mathcal{L}$, counting respectively the number of reciprocal lengths $\ell_j^{-1}$ or of frequencies $f$ not exceeding $x > 0$: $N_{\mathcal{L}} (x) = 
\sum_{\ell_j^{-1} \leq x} 1$ and $N_{\nu} (x) = \sum_{f \in \sigma (\mathcal{L})}1 $, where both the lengths and the frequencies (essentially, the square roots of the eigenvalues) are counted according to their multiplicities. (See [Lap-vFr3, Thm. 1.21].)\\

\tab For the reader who is not that familiar with analytic number theory, we next state some of the basic properties of $\zeta = \zeta(s)$ used throughout this paper, most often implicitly. The Riemann zeta function has a meromorphic extension to all of $\mathbb{C}$, still denoted by $\zeta = \zeta(s)$. Furthermore, for all $s \in \mbc$ such that $\Res > 1$, $\zeta(s)$ is given by the classic {\em Dirichlet series} $\sum_{n=1}^\infty n^{-s}$ (used to prove Equation \eqref{2.8}) and {\em Euler product} $\prod_{p \in \mathcal{P}} (1-p^{-s})^{-1}$, where $\mathcal{P}$ denotes the set of prime numbers. It follows that $\zeta(s) \neq 0$ for $\Res > 1$. Therefore, in light of the {\em functional equation} satisfied by $\zeta$ (namely, $\xi(s) = \xi (1-s)$, where $\xi (s) := \pi^{-s/2} \Gamma (s/2) \zeta (s)$ denotes the {\em completed} or {\em global} Riemann zeta function and $\Gamma$ is the classic gamma function), it follows that $\zeta (s) \neq 0$ for $\Res < 0$, except at the {\em trivial zeros} of $\zeta$, located at $s= -2, -4, -6, \cdots .$ Also, according to Hadamard's theorem (and the functional equation), we have $\zeta(s) \neq 0$ for $\Res = 1$ (and hence also for $\Res = 0$). Therefore, the other zeros of $\zeta$ (called the {\em critical} or {\em nontrivial zeros}) must lie in the (open) critical strip $0 < \Res < 1$. See, e.g., [Edw, Ing, KarVor, Pat, Ti] for these and other key properties of $\zeta$.

\subsection{Complex dimensions of fractal strings, via explicit formulas.}

Let $\mathcal{L} = \{\ell_j \}_{j=1}^\infty$ be a given fractal string and $\zeta_{\mathcal{L}} = \zeta_{\mathcal{L}} (s)$ denote its geometric zeta function, given in \S 2$(a)$ for $\Res > D$, where $D = D_{\mathcal{L}}$ is the Minkowski dimension of $\mathcal{L}.$ Briefly, the {\em complex dimensions} of $\mathcal{L}$ (or of $\partial \Omega$, where $\Omega$ is any geometric representation of $\mathcal{L}$) are the poles of the meromorphic continuation of $\zeta_{\mathcal{L}}$. More precisely, let $U \subset \mathbb{C}$ be a domain to which $\zeta_{\mcl}$ can be meromorphically continued. Then, the {\em visible complex dimensions} of $\zeta_{\mcl}$ (in $U$) are the poles of $\zeta_{\mcl}$ in $U$. If $U = \mathbb{C}$ or when no ambiguity may arise, we just call them the {\em complex dimensions} of $\mcl$. We denote by $\mathcal{D} = \mathcal{D}_{\mcl}$ the {\em set of complex dimensions} of $\mcl$. In light of the definition of the abscissa of convergence of $\zeta_{\mcl}$ recalled in \S 2$(a)$, $\zeta_{\mcl}$ is holomorphic for $\Res > D$. In fact, it can be shown  (using well-known results about Dirichlet series, see [Ser, \S VI.2.3]) that $\zeta_{\mcl} (s) \rightarrow + \infty$ as $s \rightarrow D^+, s \in \mathbb{R}$. Therefore, $\{ \Res > D \}$ is the largest open right half-plane (of the form $\{\Res > \alpha \}$, for some $\alpha \in \mathbb{R} \cup \{\pm \infty \}$) to which $\zeta_{\mcl}$ can be holomorphically continued. (Furthermore, $\{\Res > D \}$ is the largest such half-plane in which the series $\sum_{j=1}^\infty \ell_j^s$ is convergent or equivalently here, absolutely convergent; see, e.g., [Ser, \S VI.3].) It follows, in particular, that we must always have $\mathcal{D} \subseteq \{\Res \leq D \}$. If, in addition, $\zeta_{\mcl}$ can be meromorphically extended to an open neighborhood of $D$, then $D \in \mathcal{D}$ and $D = \max \{\Res : \omega \in \mathcal{D} \}.$ Here and thereafter, we (often) use the short-hand notation $\{\Res > \alpha \} := \{s \in \mathbb{C}: \Res > \alpha \}.$ Similarly, if $\alpha \in  \mathbb{R}, \{\Res = \alpha \}$ stands for the vertical line $\{s \in \mathbb{C} : \Res = \alpha \}$; and analogously for the vertical strip $\{\alpha \leq \Res \leq \beta \},$ with $\alpha \leq \beta$. \\

\tab Since we do not aim here at giving a full description of the theory of complex dimensions, for which we refer to [Lap-vFr1--3] (and especially, to the book [Lap-vFr3]), we will simply provide here a few examples and results. We note that in the present exposition, we reverse the chronological order since the results from [LapPom1--2] and [LapMai1--2] to be presented in \S3 below were obtained before the rigorous theory of complex dimensions was developed by the author and Machiel van Frankenhuijsen in [Lap-vFr1--3], and in fact (along with the work in [Lap1--3, HeLap], in particular), provided a natural heuristic and mathematical motivation for it. \\

\tab Let $\mcl = CS$ be the {\em Cantor string}; it is associated with the bounded open set $\Omega \subset \mathbb{R}$ defined as the complement in $[0,1]$ of the classic (ternary) Cantor set $C$. Note that $\Omega$ is merely the union of the sequence of disjoint open intervals (the ``middle thirds'') which are deleted in the standard construction of $C$. So that $\mcl$ is given by the sequence of lengths $\frac{1}{3}, \frac{1}{9}, \frac{1}{9}, \frac{1}{27}, \frac{1}{27}, \frac{1}{27}, \frac{1}{27}, \cdots,$ or $3^{-n}$ repeated with multiplicity $2^{n-1}$ for $n=1,2,3, \cdots$. Hence, $\zeta_{\mcl}$ has a meromorphic continuation to all of $\mathbb{C}$ given by $\zeta_{\mcl} (s) = 3^{-s}/(1-2.3^{-s}),$ for all $s \in \mathbb{C}.$ It follows that the set of complex dimensions of $C$ (or, equivalently, of the Cantor set $C = \partial \Omega$) is given by $\mathcal{D}_{CS} = \{D + in{\bf p} : n \in \mathbb{Z} \},$ where $i := \sqrt{-1}, D = \log_3 2$ is the Minkowski dimension of $CS$ (or, equivalently, of $C$) and ${\bf p} := 2 \pi/ \log_3$ is the {\em oscillatory period} of $CS$. Here and in the sequel, we let $\mathbb{Z} := \{0, \pm 1, \pm 2, \cdots \}$.\\

\tab In [Lap-vFr3, Chs. 5, 6 \& 8] are given general explicit formulas expressing, for instance, the geometric and spectral counting functions $N_{\mathcal{L}} = N_\mcl (x)$ and $N_{\nu} = N_{\nu} (x)$, respectively, as well as the volume (i.e., length) $V = V(\varepsilon) := |\{x \in \Omega: dist (x, \partial \Omega) < \varepsilon \}|$ of the (inner) $\varepsilon$-tubes of a fractal string $\mcl$, in terms of the underlying complex dimensions of $\mcl$. For example, under suitable growth conditions on $\zeta_{\mcl}$, the general {\em fractal tube formula} obtained in [Lap-vFr3, Ch. 8] is of the following form (for notational simplicity, we assume here that all the poles are simple):
\begin{align}\label{2.13}
V(\varepsilon) = \sum_{\omega \in \mathcal{D}_{\mcl}} res (\zeta_{\mcl}; \omega) \frac{(2 \varepsilon)^{1-\omega}}{\omega (1- \omega)} + R(\varepsilon), 
\end{align}
where the error term $R = R (\varepsilon)$ can be explicitly estimated. Here and in the sequel, $res (\zeta_{\mcl}; \omega)$ denotes the residue of the meromorphic function $\zeta_{\mcl} = \zeta_{\mcl} (s)$ evaluated at $s = \omega$. \\

\tab For a self-similar string, we have $R (\varepsilon) \equiv 0$ in \eqref{2.13} and so the formula is {\em exact}. In particular, for the above Cantor string $CS$, we have $\varepsilon^{-(1-D)} V_{CS} = G (\log_3 \varepsilon^{-1}) - 2 \varepsilon^D,$ where $G$ is a nonconstant periodic function (of period 1) which is bounded away from 0 and infinity. Therefore, the Cantor string $CS$ (and hence also, the Cantor set $C$) is not Minkowski measurable (as was first shown in \cite{LapPom2} via a direct computation) and its lower and upper Minkowski content are respectively given by $\mathcal{M}_* = \min G > 0$ and $\mathcal{M}^* = \max G < \infty$ (with explicit values first calculated in \cite{LapPom2}, see also [Lap-vFr3, Eq. (1.12)]). \\

\tab At the level of the counting functions $N_{CS}$ and $N_{\nu} = N_{\nu, CS}$, we have the following explicit formulas, still for the Cantor string $CS$: 
\[ N_{CS} (x) = \frac{1}{2 \log 3} \sum_{n \in \mathbb{Z}} \frac{x^{D + in{\bf p}}}{D + in{\bf p}} -1 \]
and  
\[ N_{\nu, CS} (x) = x + \frac{1}{2 \log 3} \sum_{n \in \mathbb{Z}} \zeta (D + in{\bf p}) \frac{x^{D + in{\bf p}}}{D + in{\bf p}}. \] 
(See [Lap-vFr3, \S 1.2.2 \& \S 6.4.3].) The first explicit formula was obtained by Riemann \cite{Rie} in 1858 and is certainly one of the most beautiful formulas in mathematics. It expresses the counting function of the primes in terms of the zeros and the (single) pole of $\zeta$ (i.e., in terms of the poles of $-\zeta'/\zeta$). A modern version of Riemann's explicit formula can be stated as follows: Let $\mathcal{N} (x) := \sum_{p^n \leq x} \frac{1}{n}$ be the function counting all of the prime powers $p^n \leq x$ with a weight $\frac{1}{n}$. Then 
\[ \mathcal{N} (x) = Li (x) - \sum_{\rho} Li (x^{\rho}) + \int_x^{+ \infty} \frac{1}{x^2 - 1} \frac{dx}{x \log x} - \log 2,\] 
where the sum is taken over the zeros $\rho$ of $\zeta$ and $Li (x) := \int_0^x \frac{dt}{\log t}$ is the integral logarithm. (See, e.g., [Edw, \S 1.16], [Lap-vFr3, pp.  3 \& 141], \cite{Ing}, \cite{Pat} or \cite{Ti}.) The explicit formulas obtained in [Lap-vFr3] extend this explicit formula to a general setting where the corresponding zeta function does not typically satisfy a functional equation or have an Euler product representation. (See [Lap-vFr3, Chs. 5--11].) \\

\tab From our present point of view, an important application of the general explicit formulas of [Lap-vFr3] is that if we write $\eta_{\mcl} = \sum_{j=1}^\infty \delta_{\ell_j^{-1}}$ and $\nu_{\mcl} = \sum_{f \in \sigma (\mcl)} \delta_f$, where $\delta_x$ denotes the Dirac measure (or distribution) at $x$, then (under suitable assumptions and if the complex dimensions are simple, see [Lap-vFr3, \S 6.3.1]) we have the distributional identities: $\eta_{\mcl} = \sum_{\omega \in \mathcal{D}_{\mcl}} res (\zeta_{\mcl}; \omega) x^{\omega -1}$ and $\nu_{\mcl} = \zeta_{\mcl} (1) + \sum_{\omega \in \mathcal{D}_{\mcl}} res (\zeta_{\mcl}; \omega) \zeta (\omega) x^{\omega -1},$ where we have neglected the error term. In [Lap-vFr3, \S 6.3.1], these formulas are respectively called the {\em density of geometric states} and the {\em density of spectral states.} At the level of the measures, the {\em spectral operator}, which will play a key role in the sequel (see, especially, \S4 and \S5) can be intuitively interpreted (as in [Lap-vFr3, \S 6.3.2]) as the operator sending the geometry (represented by $\eta_{\mcl}$) onto the spectrum (represented by $\nu_{\mathcal{L}}$) of a fractal string $\mcl$. In light of the above explicit formulas for $\eta_{\mcl}$ and $\nu_{\mcl}$, it amounts (essentially) to multiplying by $\zeta(\omega)$ the local term $res (\zeta_{\mcl}; \omega) x^{\omega -1}$ associated with each complex dimension $\omega \in \mathcal{D}_{\mcl}$. In \S 4$(a)$, we will discuss another heuristic interpretation of the spectral operator, based on [Lap-vFr3, \S 6.3.2] and given at the level of the counting functions $N_{\mcl}$ and $N_{\nu}$, while in \S 4$(b)$, we will give a rigorous definition of the spectral operator (obtained in [HerLap1--5]).

\begin{remark}\label{R2.1}
$(a)$ At a fundamental level, the theory of complex dimensions developed in $[$Lap-vFr1--3$]$ is a theory of geometric, spectral, dynamical, or arithmetic {\em oscillations} $($or ``vibrations''$)$. More concretely, as is apparent in the explicit formulas discussed in this subsection $($see, e.g., Eq. \eqref{2.13}$)$, the {\em real} $($resp., {\em imaginary}$)$ parts of the complex dimensions are associated with the {\em amplitudes} $($resp., {\em frequencies}$)$ of these oscillations $($viewed at the geometric level, for example, as waves propagating in the space of ``scales''$)$. \\

\tab $(b)$ In $[$Lap-vFr1--3$]$, an object is said to be {\em fractal} if its associated zeta function has at least one {\em nonreal} complex dimension $($i.e., pole$)$, with positive real part. Accordingly, all self-similar fractal strings are fractals. Furthermore, the Cantor curve $($or ``devil's staircase''$)$ is fractal in this new sense, whereas $($contrary to everyone's geometric intuition$)$, it is not fractal in the sense of Mandelbrot's definition $[$Man$]$ $($according to which the Hausdorff dimension must be strictly greater than the topological dimension$)$. We refer to the definition of fractality provided in $[$Lap-vFr3, \S 12.1 \& 12.2$].$ We note that in $[$Lap-vFr3, \S 13.4.3$]$, the above definition of fractality is extended to allow for the zeta function to have a natural boundary along some suitable curve $($i.e., not to have a meromorphic continuation beyond that curve$);$ such objects are now called {\em hyperfractals} in $[$LapRa\u Zu1--5$]$. Recently, in $[$LapRa\u Zu1--2$]$, the authors have constructed {\em maximally hyperfractal} strings $($and also compact subsets of $\mathbb{R}^N$, for any $N \geq 1)$, in the sense that the geometric $($or fractal$)$ zeta functions have singularities at every point of the vertical line $\{\Res = D \}$. \\

\tab $(c)$ The mathematical theory of complex dimensions of fractal strings has many applications to a variety of subjects, including fractal geometry, spectral geometry, number theory, arithmetic geometry, geometric measure theory, probability theory, dynamical systems, and mathematical physics. See, for example, $[$ChrIvLap, EllLapMaRo, HamLap, Fal2, HeLap, HerLap1--5, LalLap1--2, Lap1--5 \& Lap7--9, LapL\' eRo, LapLu1--3, LapLu-vFr1--2, LapMai1--2, LapNe, LapPe1--3, LapPeWi1--2, LapPom1--3, LapRa\u Zu1--5, LapRo, LapRo\u Zu, L\' eMen, MorSepVi, Pe, PeWi, Ra, Tep, \u Zu$]$, along with the three monographs $[$Lap-vFr1--3$]$ $($and the relevant references therein$)$ and the author's book $[$Lap6$]$. We note that recent developments in the theory are described in $[$Lap-vFr3, Ch. 13$]$, including a first attempt at a higher-dimensional theory of complex dimensions for the special case of fractal sprays $($in the sense of $[$LapPom3$])$ and self-similar tilings $($see $[$Lap-vFr3, \S 13.1$]$, based on $[$LapPe2--3, LapPeWi1--2, PeWi$])$, $p$-adic fractal strings and associated nonarchimedean fractal tube formulas $($see $[$Lap-vFr3, \S 13.2$]$, based on $[$LapLu1--3, LapLu-vFr1--2$]),$ multifractal zeta functions and their ``tapestries'' of complex dimensions $($see $[$Lap-vFr3, \S 13.3$],$ based on $[$LapRo, LapL\' eRo, EllLapMaRo$])$, random fractal strings $($such as stochastically self-similar strings and the zero-set of Brownian motion$)$ and their spectra $($see $[$Lap-vFr3, \S 13.4$]$, based on $[$HamLap$])$, as well as a new approach to the Riemann hypothesis based on a conjectural Riemann flow of fractal membranes $($i.e., quantized fractal strings$)$ and correspondingly flows of zeta functions $($or ``partition functions''$)$ and of the associated zeros $($see $[$Lap-vFr3, \S 13.5$]$, which gives a brief overview of the aforementioned book $[$Lap6$],$ {\em In Search of the Riemann Zeros}$)$.\\

\tab $(d)$ Recently, in $[$LapRa\u Zu1--5$],$ the theory of complex dimensions of fractal strings developed in $[$Lap-vFr1--3$]$ has been extended to any bounded subset of Euclidean space $\mathbb{R}^N$ $($and even to ``relative fractal drums'' of $\mathbb{R}^N)$ for any $N \geq 1$, via the use of new ``fractal zeta functions'', namely, the ``distance zeta function'' (introduced by the author in 2009$)$ and the closely related ``tube zeta function''. For a comprehensive exposition of the resulting theory, we refer to the forthcoming book $[$LapRa\u Zu1$]$, along with the papers $[$LapRa\u Zu2--4$]$ and survey articles $[$LapRa\u Zu5,6$].$
\end{remark}

\section{Direct and Inverse Spectral Problems for Fractal Strings}

In the present section, we study direct (\S 3$(a)$) and inverse (\S 3$(b)$) spectral problems for fractal strings, based on the work in [LapPom1,2] and [LapMai1,2], respectively, and exploring in the one-dimensional case some of the ramifications of the work in \cite{Lap1} on a partial resolution of the Weyl--Berry conjecture for fractal drums (see the discussion surrounding Eq. \eqref{3.1} below); see also \cite{Lap7} for a recent and more detailed survey on this topic. In the process, we will be led naturally to establishing close connections between these problems and the values of the Riemann zeta function $\zeta = \zeta(s)$ in the {\em critical interval} $0 < s < 1$ (see \S 3$(a)$) as well as with the {\em critical zeros} (see \S 3$(b)$); that is, the zeros of $\zeta (s)$ located in the critical strip $0 < \Res < 1$.

\subsection{Minkowski measurability of fractal strings and spectral asymptotics with a monotonic second term.}

In \cite{LapPom2} (announced in \cite{LapPom1}), the author and Carl Pomerance have established  (in the case of fractal strings, that is, in dimension one) the modified Weyl--Berry conjecture (MWB conjecture) formulated by the author in \cite{Lap1} and according to which if {\em a fractal string} $\mcl = \{\ell_j \}_{j=1}^\infty$ {\em of Minkowski dimension} $D \in (0,1)$ is Minkowski measurable {\em with Minkowski content} $\mathcal{M}$ (see the discussion at the end of \S 2$(a)$ above), {\em then its spectral} (or {\em frequency}) {\em counting function} $N_\nu = N_{\nu, \mcl}$ {\em admits a} monotonic {\em asymptotic second term, of the form} $-c_D \mathcal{M} x^D$, {\em where} $c_D$ {\em is a positive constant depending only on} $D$: 

\begin{align}\label{3.1}
N_\nu (x) = W(x) - c_D \mathcal{M} x^D + o (x^D), \text{ as } x \rightarrow + \infty,
\end{align}
where the Weyl (or leading) term $W$ is given by $W(x) := |\Omega| x,$ with $|\Omega| = \sum_{j=1}^\infty \ell_j$ being the ``volume'' (or really, the ``length'') of any geometric realization $\Omega$ of $\mcl$.

\begin{theorem}[Resolution of the MWB conjecture for fractal strings, \cite{LapPom2}]\label{T3.1}
If $\mcl$ is a Minkowski measurable fractal string of dimension $D \in (0,1)$, then \eqref{3.1} holds, with the Weyl term given by $W(x) := |\Omega| x$ as above and the positive constant $c_D$ given by $c_D := (1-D) 2^{D-1} (-\zeta(D))$.
\end{theorem}  

\tab Note that $c_D > 0$ because the Riemann zeta function is negative in the {\em critical interval} $(0,1)$ (see, e.g., \cite{Ti}): $\zeta (D) < 0$ since $D \in (0,1)$. The proof of Thm. 3.1 provided in \cite{LapPom2} is given in two different steps. The first step consists in proving the following key characterization of Minkowski measurability for fractal strings (i.e., in one dimension), due to \cite{LapPom2} and also of independent geometric interest:  \\

\tab {\em The fractal string} $\mcl = \{\ell_j\}_{j=1}^\infty$ ({\em of Minkowski dimension} $D \in (0,1))$ {\em is Minkowski measurable if and only if} $\ell_j \sim Lj^{-1/D}$ {\em as} $j \rightarrow \infty$, {\em for some} $L \in (0, + \infty)$; i.e., iff $L := \lim_{j \rightarrow \infty} \ell_j \cdot j^{1/D}$ {\em exists in} $(0, + \infty)$. {\em In that case, its Minkowski content $\mathcal{M}$ is given by} $\mathcal{M} = 2^{1-D} L^D/(1-D)$. \\

\tab The second step of the proof of Thm. \ref{T3.1}, also due to \cite{LapPom2}, consists in first observing that (since the intervals or substrings of which the fractal string $\Omega$ is comprised are vibrating independently of one another) $N_\nu (x) = \sum_{j=1}^\infty [\ell_j x]$, where $[y]$ is the integer part of $y \in \mathbb{R}$, and then establishing the following ({\em monotonic}) {\em two-term asymptotic expansion holds}: If $\ell_j \sim L j^{-1/D}$ as $j \rightarrow \infty$, then $\sum_{j=1}^\infty [\ell_j x] = (\sum_{j=1}^\infty \ell_j)x + \zeta (D) x^D + o (x^D),$ as $x \rightarrow + \infty$. (Observe that the leading term in this formula coincides  with the Weyl term, $W(x) = |\Omega| x$.) In the proof given in [LapPom1,2], one uses, in particular, the following form of the analytic continuation of $\zeta = \zeta(s)$ to the open right half-plane $\{\Res > 0 \}$ (see, e.g., [Ser, \S VI.3.2]): $\zeta(s) = \frac{1}{s-1} + \int_0^{+ \infty} ([t]^{-s} - t^{-s}) dt,$ for $\Res > 0$. Note that on the right-hand side of this identity, the first term yields the single, simple pole of $\zeta$ at $s=1$, while the second term defines a holomorphic function of $s$ for $\Res > 0$. This identity also plays an important role in [HerLap1,5]; see \S 4$(c)$ below.\label{F3.1} We mention that later proofs of a part of Step 1 above were provided in \cite{Fal2} and in \cite{RatWi}, using different techniques from the theory of dynamical systems or from geometric measure theory, respectively.

\begin{remark}\label{R3.2}
$(a)$ The original Weyl conjecture was stated $($in 1912$)$ for ``nice'' smooth $($or piecewise smooth$)$ bounded domains $\Omega \subseteq \mathbb{R}^N (N \geq 1)$ in $[$We$],$ where the mathematician Hermann Weyl also obtained his classic formula  $($see, e.g., $[$CouHil, Kac$]$ or $[$ReSi, vol. IV$])$ for the leading spectral asymptotics of the Dirichlet $($or Neumann$)$ Laplacian on $\Omega$ in the case of a $($piecewise$)$ smooth boundary $\partial \Omega$. The original and stimulating Berry conjecture was stated $($in the late 1970s$)$ by the physicist Michael Berry in $[$Berr1,2$]$ for bounded open sets $\Omega \subseteq \mathbb{R}^N$ with ``fractal'' boundary. It was expressed in terms of the Hausdorff dimension and measure of the boundary, and in that form, was disproved in $[$BroCar$];$ see also $[$Lap1$],$ for general mathematical reasons and a simple family of counterexamples, valid in any dimension. A partial resolution of the Weyl--Berry conjecture for fractal drums was obtained by the author in $[$Lap1$],$ where it was proved that for the Dirichlet $($as well as, under suitable hypotheses, the Neumann$)$ Laplacian on a bounded open set $\Omega \subseteq \mathbb{R}^N$ $($for $N \geq 1)$ with boundary $\partial \Omega$ satisfying $D = D (\partial \Omega) \in (N-1, N)$ and $\mathcal{M}^* < \infty$, we have
\begin{align}\label{3.6}
N_\nu (x) = W(x) + O(x^D), \ \textnormal{ as } \ x \rightarrow +\infty,
\end{align} 
where $W(x) := \mathcal{C}_N |\Omega|_N x^N $ is the $N$-dimensional Weyl term and $|\Omega|_N $ is the $N$-dimensional volume of $\Omega$. Here, $D$ notes the (upper) Minkowski dimension of $\partial \Omega$ and, as noted in $[$Lap1$]$, we always have $D \in [N-1, N]$. The error estimate in \eqref{3.6} was also shown to be sharp for every $N \geq 1$ and every $D \in (N-1, N);$ see $[$Lap1, Examples 5.1 and 5.1'$]$. In the nonfractal case when $D = N-1$, the error term in the counterpart of \eqref{3.6} then involves a logarithm term; see $[$Lap1$]$ and the earlier work by G. M\' etivier in $[$Met$].$ For further results concerning the $($modified$)$ Weyl--Berry conjecture and discussion of its many physical applications $[$Berr1--2$],$ we refer the interested reader to $[$Lap1--3$],$ $[$Lap-vFr3, \S 12.5 and App. B$]$, $[$LapRa\u Zu1,5$]$ and $[$Lap7, \S 4.1$],$ along with the relevant references therein, including $[$BroCar, FlVas, Ger, GerSchm, HamLap, HeLap, LapPom1--3, LapMai1--2, MolVai, Sch$].$ $($We note that the analog of the MWB conjecture is not true, in general, when $N \geq 2$, although no counterexample seems to be known for simply connected planar domains$;$ see $[$FlVas$]$ and, especially, $[$LapPom3$].)$ For the case of ``drums with fractal membrane" $($Laplacians on fractals, $[$Ki$]$, rather than on open sets with fractal boundary, $[$Lap1--3$]),$ we refer to $[$KiLap1$]$ and $[$Lap3$]$ along with, e.g., $[$RamTo, Ham, KiLap2, Sab, ChrIvLap, LapSar, Tep, LalLap1--2$]$ and the relevant references therein. \\

\tab $(b)$ It is shown in $[$LapPom2$]$ that a fractal string $($of dimension $D \in (0,1))$ is Minkowski nondegenerate $($i.e., $0 < \mathcal{M}_* (\leq) \mathcal{M}^* < \infty)$ if and only if $\ell_j \approx j^{-1/D}$ as $j \rightarrow \infty$ and if and only if $\varphi (x) := | N_\nu (x) - W(x) | \approx x^D$ as $x \rightarrow + \infty$, where the symbol $\approx$ means that we have two-sided error estimates$;$ e.g., $\alpha x^D \leq \varphi (x) \leq \beta x^D$, for some constants $\alpha, \beta > 0$. We will next see $($in \S 3$(b)$ just below$)$ that the situation is very different when ``Minkowski nondegeneracy'' is replaced by ``Minkowski measurability''. 
\end{remark}

\subsection{The sound of fractal strings and the Riemann zeros.}

In \S 3$(a)$ just above, we have studied a {\em direct spectral problem for fractal strings} since in Thm. \ref{T3.1} based on \cite{LapPom2} (as well as in the corresponding MWB conjecture from \cite{Lap1}), we assume something about the geometry of the fractal string $\mcl$ (namely, the Minkowski measurability of $\mcl$) and deduce some information about the spectrum of $\mcl$ (namely, that $N_\nu (x)$ admits a monotonic asymptotic second term). Conversely, it is natural to consider (as was done in [LapMai1,2]) the following associated inverse spectral problem: \\

(ISP)$_D$ {\em If the fractal string} ({\em of Minkowski dimension} $D \in (0,1)$) {\em is such that its spectral counting function} $N_\nu (x)$ {\em admits a} monotonic {\em asymptotic second term as} $x \rightarrow + \infty$, {\em proportional to} $x^D$ (i.e., of the form $\mathcal{C} x^D$, where the nonzero constant $\mathcal{C}$ depends only on $D$ and $\mathcal{L}$, compare with \eqref{3.1} above), {\em is it true that} $\mcl$ {\em is necessarily Minkowski measurable}? \\

\tab As was shown by the author and Helmut Maier in \cite{LapMai2} (announced in \cite{LapMai1}), it turns out that this question (\`a la Mark Kac, but of a very different nature than in \cite{Kac}) ``{\em Can one hear the shape of a fractal string}?'' is intimately related to the existence of critical zeros of $\zeta = \zeta(s)$ on the vertical line $\{\Res = D \}$, and therefore to the Riemann hypothesis. {\em In the sequel, we will say that the inverse spectral problem} (ISP)$_D$ {\em has an affirmative answer for a given value of $D \in (0,1)$ if the above fractal strings version of Kac's question has an affirmative answer for all fractal strings of dimension $D$.} Equivalently, (ISP)$_D$ has a {\em negative answer} for some given $D \in (0,1)$ if there exists a non Minkowski measurable fractal string $\mcl$ of (Minkowski) dimension $D$ such that its spectral counting function $N_\nu = N_{\nu, \mcl}$ has a monotonic asymptotic second term (proportional to $x^D$). Note, however, that in light of the result of \cite{LapPom2} discussed in Remark \ref{R3.2}($b$), such a fractal string $\mcl$ must be Minkowski nondegenerate.

\begin{theorem}[Riemann zeros and the inverse problem (ISP)$_D$, \cite{LapMai2}]\label{T3.4}
Fix $D \in (0,1).$ Then, the inverse spectral problem $($ISP$)_D$ has an affirmative answer for this value of $D$ if and only if $\zeta = \zeta(s)$ does not have any zeros on the vertical line $\{\Res = D \}$. 
\end{theorem}

\tab It follows from Thm. \ref{T3.4} that in the midfractal case when $D= 1/2$, the inverse problem (ISP)$_D$ has a negative answer since $\zeta$ has a zero on the critical line $\{\Res =1/2 \}$. (Actually, $\zeta(s)$ vanishes infinitely often for $\Res = 1/2$ but this result, due to Hardy, see \cite{Ti}, is not needed here. One zero, and its complex conjugate, which is also a zero of $\zeta$ since $\zeta (\overline{s}) = \overline{\zeta (s)}$, suffices.) Furthermore, any counterexample to the Riemann hypothesis would show that (ISP)$_D$ fails to have an affirmative answer for some $D \neq 1/2$. (None is expected, however.) Hence, we obtain the following result, which is really a corollary of Thm. \ref{T3.4} (in light of the functional equation satisfied by $\zeta = \zeta(s)$, which exchanges $s$ and $1-s$) but that we state as a theorem because it is the central result of [LapMai1,2] and is a key motivation for several of the main results stated in the later sections (\S4 and \S5).

\begin{theorem}[Riemann hypothesis and inverse spectral problems for fractal strings, \cite{LapMai2}]\label{T3.5}
The inverse spectral problem $($ISP$)_D$ has an affirmative answer for all $D \in (0,1)$ other than $D = 1/2$ $[$or, equivalently, for all $D \in (0, 1/2)$ or else, for all $D \in (1/2, 1)$, respectively$]$ if and only if the Riemann hypothesis is true.
\end{theorem}

\tab Thm. \ref{T3.5} provides a spectral and geometric reformulation of the Riemann hypothesis. For the purpose of this paper, we will refer to Thm. \ref{T3.5} (as well as to Thm. \ref{T4.7}, a corresponding operator theoretic version of Thm. \ref{T3.5} obtained in [HerLap1--3] and discussed in \S 4$(d)$ below) as a {\em symmetric criterion for the Riemann hypothesis} (RH). Obtaining an {\em unsymmetric criterion for} RH will be the object of \S5.

\begin{remark}\label{R3.6}
The proof of Thm. \ref{T3.5} relies on the Wiener--Ikehara Tauberian theorem $[$Pos$]$ $($for one of the implications$)$ and $($for the reverse implication) on the intuition of the notion of complex dimension which, at the time, was still conjectural $($see $[$LapPom1--2, Lap2--3$])$. On the other hand, Thm. \ref{T3.5}, along with results and conjectures in $[$Lap1--3$],$ $[$LapPom1--3$],$ $[$KiLap1$]$ and $[$HeLap$]$, in particular, provided some of the key motivations for developing in $[$Lap-vFr1--3$]$ a rigorous theory of complex $($fractal$)$ dimensions via $($generalized$)$ explicit formulas and geometric zeta functions. We also refer to $[$Lap-vFr1--3$]$, especially $[$Lap-vFr3, Ch. 9$]$, where Thm. \ref{T3.5} is reformulated in terms of explicit formulas and extended to all arithmetic zeta functions for which the generalized Riemann hypothesis (GRH), is expected to hold. For information about GRH, see, e.g., $[$ParSh$]$, $[$Sarn$]$, $[$Lap6, Apps. B, C \& E$]$ or $[$Lap-vFr3, App. A$]$.
\end{remark}

\section{The Spectral Operator as a Quantized Riemann Zeta Function}

In this section, we give a brief overview of the work of Hafedh Herichi and the author on aspects of quantized number theory. This work is presented in the series of research and survey articles [HerLap2--5], and described in detail in the forthcoming book, \cite{HerLap1}, titled {\em Quantized Number Theory, Fractal Strings and the Riemann Hypothesis}: {\em From Spectral Operators to Phase Transitions and Universality.} More specifically, in [HerLap1--5], the main object of investigation is the spectral operator, sending the geometry onto the spectrum of a fractal string. Originally seen as acting on the space of (generalized) fractal strings, it is now viewed mathematically (and physically) as a suitable quantization of the Riemann zeta function: $\mfa = \zeta (\partial)$, where $\partial$ is a suitable version of the differentiation operator $d/dt$, called the infinitesimal shift of the real line and acting on the Hilbert space $\mathbb{H}_c = L^2 (\mathbb{R}, e^{-2ct} dt)$, a weighted $L^2$-space; see \S 4$(b)$. Accordingly, the spectral operator $\mfa = \mfa_c$ depends on a parameter $c$ which can be thought of heuristically as providing an upper bound for the Minkowski dimensions of the fractal strings on which the operator $\mfa$ acts. In fact, as we shall see in \S 4$(a)$, it is more convenient to replace the fractal strings themselves by their associated (geometric and spectral) counting functions. As it turns out, the spectral operator, $\mfa = \zeta (\partial)$, which was originally introduced at the semi-heuristic level in [Lap-vFr2, \S 6.3.1 \& \S 6.3.2] (see also [Lap-vFr3, \S 6.3.1 \& \S 6.3.2]) satisfies (at the operator theoretic level) most of the properties of the classic Riemann zeta function $\zeta = \zeta (s)$, including a quantized Dirichlet series, a quantized Euler product and an operator-valued ``analytic continuation''. (See \S 4$(c)$, along with [HerLap1, Ch. 7] and [HerLap5].) One of the key motivations of the work in [HerLap1--5] is to obtain a rigorous functional analytic version of the work of the author and H. Maier in [LapMai2] on inverse spectral problems for fractal strings and the Riemann hypothesis briefly described in \S3$(b)$ above. This goal is achieved (see \S 4$(d)$ below) by studying the quasi-invertibility of the spectral operator as a function of the parameter $c$ in the critical interval $(0,1)$, which is the natural range of possible dimensions of fractal strings, leaving aside the least and most fractal cases. In this light, the Riemann hypothesis is true if and only if the midfractal case $c = 1/2$ is the only exception to the quasi-invertibility of $\mfa$; see \S 4$(d)$. 

\subsection{Heuristic spectral operator.} As we have seen in \S 2$(c)$, the analog (at the level of the counting functions) of the factorization formula \eqref{2.8}, connecting the geometric and the spectral zeta functions of a fractal string $\mcl$, via the Riemann zeta function $\zeta$, is provided by the following formula, connecting the geometric and spectral counting functions $N_{\mcl}$ and $N_\nu$, respectively, of a given fractal string $\mcl$: $N_\nu (x) = \sum_{n=1}^\infty N_\mcl (x/n),$ for all $x > 0$. (Note that for a fixed $x > 0$, only finitely many terms contribute to this sum; however, the number of these terms tends to infinity as $x \rightarrow + \infty$.) The (heuristic) spectral operator, at the level of the counting functions, is then given by the map $g(x) = N_\mcl (x) \mapsto N_\nu (g) (x) = \sum_{n=1}^\infty g (x/n).$ Hence, it can be thought of as sending the geometry onto the spectrum of a fractal string $\mcl$. At an even more fundamental level, and using the notation of \S 2$(d)$, the heuristic spectral operator can be thought of as being the map $\eta := \sum_{j=1}^\infty \delta_{\ell_j^{-1}} \mapsto \nu = \nu (\eta) := \sum_{f \in \sigma (\mcl)} \delta_f$, where $\sigma (\mcl)$ denotes the (frequency) spectrum of $\mcl$ (given at the beginning of \S 2$(c)$) and the measures $\eta$ and $\nu$ respectively represent the geometry and the spectrum of $\mcl$, viewed as generalized fractal strings (in the sense of [Lap-vFr3, Ch. 4]). (For a brief discussion of the spectral operator expressed in terms of the complex dimensions of the fractal string $\mcl$, see the text immediately preceding Remark \ref{R2.1} in \S 2$(d)$ about the densities of geometric and of spectral states of $\mcl$; see also [Lap-vFr2--3, \S 6.3.1].)\\

\tab The spectral operator was first introduced (in this semi-heuristic context) in [Lap-vFr2, \S 6.3.2] (see also [Lap-vFr3, \S 6.3.2]). The above version is referred to in [HerLap1--5] as the multiplicative version of the (heuristic) spectral operator. From now on, we will only work with the (equivalent) additive version, which we next describe. \\

\tab Viewed additively (that is, after having made the change of variable $x= e^t, t = \log x$, with $x > 0$ and $t \in \mathbb{R}$), we obtain the additive spectral operator $\mfa$, to be briefly referred to henceforth as the (heuristic) ``{\em spectral operator}'':

\begin{align}\label{4.4}
f(t) \mapsto \mfa(f)(t) = \sum_{n=1}^\infty f (t- \log n),
\end{align}
where the functions $f= f(t)$ are viewed as functions of the variable $t \in \mathbb{R} = (-\infty, +\infty).$ (The precise mathematical setting will be specified in \S 4$(b)$.) Given a prime $p$ (i.e., $p \in \mathcal{P}$, where $\mathcal{P}$ denotes the set of all prime numbers), the associated {\em local} ({\em operator-valued}) {\em Euler factor} $\mfa_p$ is given by $f(t) \mapsto \sum_{m=0}^\infty f (t -m \log p)$. Formally, the spectral operator $\mfa$ is connected to its local Euler factors $\mfa_p$ (with $p \in \mathcal{P})$ via the following {\em Euler product representation}: $\mfa = \prod_{p \in \mathcal{P}} \mfa_p$, where the infinite product is to be understood in the sense of the composition of operators. Next, assuming for now that the function $f$ is infinitely differentiable, we formally define the differentiation operator $\partial = d/dt$ (so that $\partial f = f'$, the derivative of $f, \cdots, \partial^k f = f^{(k)}$, the $k$-th derivative of $f$, for any integer $k \geq 0$). Then, $e^{-h \partial}$ acts a shift or translation operator. Namely, $f(t-h) = (e^{-h \partial})(f)(t).$ \\

\tab Let us close this discussion by giving the heuristic motivations for the above Euler product formula and for the following Dirichlet series representation of $\mfa$, which will serve as the basis for our rigorous discussion in \S 4$(b)$ and \S 4$(c)$: $\mfa = \sum_{n=1}^\infty n^{-\partial} = \zeta (\partial)$. In light of the above discussion, we have for any prime $p \in \mathcal{P}$, $\mfa_p = \sum_{m=0}^\infty e^{-m (\log p) \partial} = \sum_{m=0}^\infty (p^{-\partial})^m = (1-p^{-\partial})^{-1}$ and so $\mfa = \sum_{n=1}^\infty e^{-(\log n)\partial} = \sum_{n=1}^\infty (n^{-\partial}) = \zeta (\partial),$ as desired. All of these formulas should be verified by applying the corresponding expression to the function $f$ and evaluating at a generic $t \in \mathbb{R}$. In actuality, the rigorous justification of these heuristic formulas requires a lot more work and is provided (under suitable assumptions) in [HerLap1--5], as we next briefly discuss in \S 4$(b)$ and \S 4$(c)$.

\subsection{The infinitesimal shift $\partial$ and the spectral operator $\mfa = \zeta(\partial)$.}

Given $c \in \mathbb{R}$, let $\mathbb{H}_c = L^2(\mathbb{R}, e^{-2ct} dt)$ be the complex Hilbert space of (complex-valued, Lebesgue measurable) square-integrable functions with respect to the absolutely continuous measure $\mu (dt) = e^{-2ct} dt$. The inner product of $\mathbb{H}_c$ is given (for $f, g \in \mathbb{H}_c$) by $<f,g>_c := \int_{-\infty}^{+ \infty} e^{-2ct} f(t) \overline{g}(t)dt$ and thus, the associated norm $||\cdot||_c$ is naturally given by $||f||_c^2 = \ <f,f>_c \ = \int_{-\infty}^{+\infty} e^{-2ct} |f(t)|^2 dt < \infty.$ \\

\tab Let $C_{abs}(\mathbb{R})$ denote the space of (locally) absolutely continuous functions on $\mathbb{R}$; see, e.g., \cite{Foll} or \cite{Ru1}. Let $\partial = \partial_c = d/dt$ denote the differentiation operator, acting on $\mathbb{H}_c$ and with domain $D(\partial) = \{f \in \mathbb{H}_c \cap C_{abs} (\mathbb{R}): f' \in \mathbb{H}_c \}$, where $f'$ denotes the distributional (or weak) derivative of $f$ (see, e.g., [Bre, JoLap, Ru2, Schw]), which (since $f$ is absolutely continuous) can be interpreted as the usual pointwise derivative of $f$ existing (Lebesgue) almost everywhere ($a.e.$, in short) on $\mathbb{R}$. Then, for every $f \in D(\partial)$, we have $\partial f := f'$. The operator $\partial$ so defined is called the {\em infinitesimal shift} of the real line. (Note that it depends on the choice of the parameters $c \in \mathbb{R}$.)

\begin{theorem}[\cite{HerLap1}]\label{T4.1}
For every $c \in \mathbb{R}$, the infinitesimal shift $\partial = \partial_c$ is an unbounded normal operator on $\mathbb{H}_c$, with spectrum $\sigma (\partial)$ given by $\sigma (\partial) = \{\Res = c \} = c + i \mathbb{R}$, the vertical line of abscissa $c$. $($Here, $i := \sqrt{-1}.)$ Furthermore, its adjoint $\partial^*$ is given by $\partial^* = 2c-\partial ;$ so that we also have $\sigma (\partial^*) = \{\Res = c \}$. In addition, neither $\partial$ {\em nor} $\partial^*$ has any eigenvalues $($hence, its point spectrum is empty$).$
\end{theorem} 

\tab Recall that given a (closed) unbounded, linear operator $L$ on a complex Hilbert space $H$, its {\em spectrum}, denoted by $\sigma (L)$, is defined as the  set of all $\lambda \in \mathbb{C}$ such that $L - \lambda I$ is not invertible, where $I$ denotes the identity operator. (In finite dimensions, this amounts to requiring that $L - \lambda I$ is not 1--1, and hence, $\lambda$ is an eigenvalue of $L$; so that $\sigma (L)$ is the finite set of eigenvalues of $L$, in that case.) The spectrum $\sigma (L)$ is always a closed subset of $\mathbb{C}$. Furthermore, $L$ is invertible if and only if $0 \notin \sigma (L)$. Moreover, an (unbounded, linear) operator $M$, with domain $D(M)$, is said to be {\em invertible} if the linear map $M: D(M) \rightarrow H$ is bijective and its set theoretic inverse $M^{-1}: H \rightarrow D(M) \subseteq H$ is bounded (which, by the closed graph theorem, is automatically true if $M$ is closed, as will always be the case in this paper). Finally, if $M$ is closed (i.e., if its graph is a closed subspace of $H \times H$ for the graph norm), then it is said to be {\em normal} if $M^*M = MM^*$ (equality between unbounded operators), where $M^*$ denotes the adjoint of $M$. For all of these notions, we refer, e.g., to [DunSch, JoLap, Kat, Ru2] and [ReSi, vol. I].\\

\tab In the next result, which is needed (in particular) to rigorously justify the formulas provided in \S 4$(a)$ (see \S 4$(c)$ below), we will implicitly use the theory of strongly continuous contraction semigroups of bounded linear operators and their infinitesimal generators (i.e., $m$-accretive operators; see, e.g., [Go, HiPh, JoLap, Kat, ReSi]).

\begin{proposition}[\cite{HerLap1}]\label{P4.2}
The semigroup $\{e^{-h \partial} \}_{h \geq 0}$ $($resp., $\{e^{h \partial} \}_{h \geq 0})$ with infinitesimal generator $\partial = \partial_c$ $($resp., $-\partial)$ is a contraction semigroup on $\mathbb{H}_c$ if $c \geq 0$ $($resp., $c \leq 0)$. Furthermore, if the roles of $\partial$ and $-\partial$ (or equivalently, of $c$ and $-c$) are interchanged, so are the roles of ``contraction'' and ``expansion'' semigroups in this statement. Moreover, for every $f \in \mathbb{H}_c$ and $h \in \mathbb{R}$, we have $||e^{-h \partial} f||_c = e^{-hc} ||f||_c$ and $(e^{-h \partial}f)(t) = f(t-h)$, \text{ for a.e.} $t \in \mathbb{R}$ $($so that $e^{-h \partial} f = f (\cdot - h)$ in $\mathbb{H}_c)$. As a result, $\partial$ is also referred to as the {\em infinitesimal shift of the real line.}
\end{proposition}

\begin{remark}\label{R4.3}
Given $c \in \mathbb{R}$, the $c$-{\em momentum operator} is given by $p_c := \frac{1}{i} \partial_c$, where $i := \sqrt{-1}$. It is self-adjoint if and only if $c=0$, in which case it coincides with the classic momentum operator $p_0$ of quantum mechanics, acting on $\mathbb{H}_0 = L^2 (\mathbb{R})$ and with spectrum $\sigma (p_0) = \mathbb{R};$ see, e.g., $[$Kat, ReSi, Sc$].$ Similarly, if we let $-\Delta_c = (p_c)^2 = -(\partial_c)^2$ denote the {\em free} $c$-{\em Hamiltonian}, then the group $\{e^{-ih \Delta_c} \}_{h \in \mathbb{R}}$ is unitary if and only if $c = 0$.
\end{remark}

\tab We can next proceed to rigorously define the {\em spectral operator} $\mfa = \mfa_c$ by the expression $\mfa_c := \zeta (\partial_c)$ or, more concisely, $\mfa := \zeta (\partial)$, in the sense of the functional calculus for unbounded normal operators (see, e.g., \cite{Ru2}). Of course, we are using here the fact that (according to Thm. \ref{T4.1} above), $\partial_c$ is an unbounded normal operator. Note that, by definition, the domain of $\mfa$ is given by $D(\mfa) = \{f \in \mathbb{H}_c : \mfa f \in \mathbb{H}_c \}.$ In the sequel, $c \ell (E)$ denotes the closure of $E \subseteq \mathbb{C}$ in $\mathbb{C}$.

\begin{theorem}[\cite{HerLap1}]\label{T4.4}
For every $c \in \mathbb{R},$ the spectral operator $\mfa = \mfa_c$ is a $($possibly unbounded$)$ normal operator on $\mathbb{H}_c$, with spectrum $\sigma (\mfa)$ given by $\sigma (\mfa_c) = c \ell (\{\zeta (s): \Res = c, s \in \mathbb{C} \})$. $($When $c = 1$, the pole of $\zeta$, we must exclude $s=1$ on the right-hand side of the above expression for $\sigma (\mfa_c).)$ 
\end{theorem}

\tab In words, Thm. \ref{T4.4} states that the spectrum of $\mfa$ coincides with the closure (i.e., the set of limit points) of the range of $\zeta$ along the vertical line $\{\Res = c \}$. This result follows from Thm. \ref{T4.1} combined with a suitable version of the spectral theorem for unbounded normal operators (see [HerLap1, App. E]) according to which [since $\partial$ does not have any eigenvalue and $\zeta$ is continuous for $c \neq 1$ (resp., meromorphic for $c=1$) in an open connected neighborhood of $\sigma (\partial) = \{\Res = c \}]$, we have that $\sigma (\mfa) = \sigma (\zeta (\partial)) = c \ell (\zeta (\sigma (\partial)))$, for $c \neq 1$. (For $c=1$, the extended spectrum is given by $\widetilde{\sigma} (\mfa) := \sigma (\mfa) \cup \{\infty \} = \zeta (\widetilde{\sigma}(\partial))$, where the meromorphic function $\zeta$ is now viewed as a continuous function from $\mathbb{C}$ to the Riemann sphere $\widetilde{\mathbb{C}} := \mathbb{C} \cup \{\infty \}$; from which the result also follows for $c=1.$)

\subsection{Justification of the definition of $\mfa$: Quantized Dirichlet series and Euler product.} The first justification for the definition of $\mfa$ given by $\mfa = \zeta(\partial)$ comes from the following result. In the sequel, $\mathcal{B}(\mathbb{H}_c)$ stands for the Banach algebra (and even, $C^*$-algebra) of bounded linear operators on $\mathbb{H}_c$, equipped with its natural norm: $||L|| := \sup \{ ||L f ||_c : f \in \mathbb{H}_c, ||f||_c \leq 1 \}$.

\begin{theorem}$([$HerLap1,5$])$.\label{T4.5}
For $c >1, \mfa = \zeta (\partial)$ is given by the following operator-valued $($or ``quantized''$)$ Dirichlet series and Euler product, respectively$:$ $\mfa = \zeta(\partial) = \sum_{n=1}^\infty n^{-\partial} = \prod_{p \in \mathcal{P}} (1- p^{-\partial})^{-1},$ where the series and the infinite product both converge in $\mathcal{B}(\mathbb{H}_c)$. Furthermore, for any integer $m \geq 1$ and $f \in \mathbb{H}_c$, we have $(m^{-\partial}) (f) (t) = f(t - \log m), \text{ for a. e. } t \in \mathbb{R}$ $($and hence, $(m^{-\partial})(f) = f (\cdot - \log m)$ in $\mathbb{H}_c)$.
\end{theorem}

The latter part of Thm. \ref{T4.5} follows from Prop. \ref{P4.2} since $m^{-\partial} = e^{-(\log m) \partial}$, in the sense of the functional calculus. Naturally, Thm. \ref{T4.5} also provides a rigorous justification of the heuristic discussion (based on [Lap-vFr2--3, \S 6.3.2]) given in \S 4$(a)$ above. The following results (also from \cite{HerLap5} and [HerLap1, Ch. 7]) go well beyond that discussion but will not be precisely or fully stated here, by necessity of concision: \\

\tab (1) ({\em Analytic continuation of} $\mfa_c = \zeta (\partial_c)$, {\em for} $c > 0$). For $c > 0, \mfa = \mfa_c$ coincides with the operator-valued ``analytic continuation'' of the quantized Dirichlet series and Euler product (defined only for $c > 1$, initially). More specifically, for a dense subspace of functions $f$ in $D(\partial)$ (and hence also in $\mathbb{H}_c$), we have for all $c > 0$, $\mfa(f) = \left(\frac{1}{\partial -1}\right) (f) + \int_0^{+ \infty} ([t]^{-\partial} (f) -t^{-\partial} (f)) dt,$ where $[t]$ is the integer part of $t$. Note that this identity is an operator theoretic version of the analytic continuation of $\zeta = \zeta(s)$ to $\{\Res > 0\}$ given in \S 3$(a)$.\\

\tab (2) ({\em Analytic continuation of} $\mathcal{A}_c =\xi (\partial_c),$ {\em for} $c \in \mathbb{R}$, {\em and quantized functional equation}). Let the {\em global spectral operator} be defined by $\mathcal{A} = \xi (\partial)$, where $\xi$ is the global (or completed) zeta function $\xi (s) := \pi^{-s/2} \Gamma (s/2) \zeta (s)$. Then, for all $c \in \mathbb{R}$, the operator-valued analytic continuation of $\mathcal{A} = \mathcal{A}_c$ is given by a formula analogous to the one in (1) just above, but in which the symmetry $s \leftrightarrow 1-s$ is more immediately manifest. (This is a quantized counterpart of the corresponding formula for $\xi (s)$; see, e.g., [Edw, Ti] or [Lap6, Eq. (2.4.10)].) One deduces from this identity an operator-valued analog (for $\mathcal{A}_c$) of the classic functional equation $\xi(s) = \xi (1-s)$. Namely, $\mathcal{A}_c = \mathcal{B}_c$, where $\mathcal{B}_c := \xi (I -\partial_c)$. In particular, $\xi (\partial_{1/2}) = \xi (\partial_{1/2}^*)$. \\

\tab (3) ({\em Inverse of} $\mfa$ {\em for} $c >1$). For $c>1$ (as in Thm. \ref{T4.5}), $\mfa = \zeta(\partial)$ is invertible (in $\mathcal{B} (\mathbb{H}_c)$), with inverse $\mfa^{-1}$ given by $\mfa^{-1} = ( 1 / \xi )(\partial)= \sum_{n=1}^\infty \mu (n) n^{-\partial}$, where the series is convergent in $\mathcal{B}(\mathbb{H}_c)$ and $\mu$ is the classic M\" obius function (see, e.g., \cite{Edw}), defined by $\mu(n) = 0$ if $n$ is not square-free and $\mu(n) = \pm 1$ depending on whether $n$ is the product of an even or odd number of primes, respectively. We leave it to the reader to obtain a similar formula, but now expressed in terms of an infinite product, using the quantized Euler product given in Thm. \ref{T4.5}. \\

\tab We hope that the above discussion provides a sufficiently convincing sample of formulas from ``{\em quantized number theory}'', in the terminology of [HerLap1--5] and [Lap8--9]. Many additional formulas can be obtained, involving either $\zeta (\partial)$ (or $\xi(\partial)$) or, more generally, $\zeta_L (\partial),$ where $\zeta_L$ is any of the classic $L$-functions of number theory and arithmetic geometry; see, e.g., [ParSh, Sarn, Lap6] and [Lap-vFr3, App. A]. Further exploration in this direction is conducted in \cite{Lap8} and in \cite{Lap9}, either in the present framework or else in a different functional analytic and operator theoretic framework (using Bergman spaces, [AtzBri, HedKorZh]); see Rem. \ref{R5.4}($b$) below.

\subsection{Quasi-invertibility of $\mfa$ and the Riemann hypothesis.} 
In \S 3$(b)$ above, we have briefly described the work of \cite{LapMai2} in which a spectral reformulation of the Riemann hypothesis was obtained; see, specifically, Thm. \ref{T3.5}. We will provide here an operator theoretic analog of this reformulation, based on the work of [HerLap1--3].\\

\tab Given $T > 0$, let  $\mfa^{(T)} := \zeta (\partial^{(T)})$ denote the $T$-{\em truncated spectral operator}, where $\partial^{(T)} := \rho^{(T)}(\partial)$ is the $T$-{\em truncated infinitesimal shift}, defined (thanks to Thm. \ref{T4.1} above) via the functional calculus for unbounded normal operators. In practice, one simply refers to $\mfa^{(T)}$ and $\partial^{(T)}$ as the truncated spectral operator and the truncated infinitesimal shift, respectively. The precise definition of the {\em cut-off function} $\rho^{(T)}$ is unimportant, provided $\rho^{(T)}$ is chosen so that the following computation of the spectrum of $\partial^{(T)}$ is justified (since $\sigma (\partial) = c + i \mathbb{R}$, by Thm. \ref{T4.1}): $\sigma (\partial^{(T)}) = c \ell (\rho^{(T)}(\sigma (\partial))) := [c-iT, c + iT].$ Since $\partial^{(T)} = \rho^{(T)} (\partial)$, the first equality follows from the spectral mapping theorem (SMT), (applied to the unbounded normal operator $\partial$ and the function $\rho^{(T)}$, see [HerLap1, App. E]), while the second equality follows from the construction of $\rho^{(T)}$. In fact, when $c \neq 1$, the only requirement about $\rho^{(T)}$ is that it be a $\mbc$-valued, continuous function on $c + i \mbr = \{\Res =c \}$ whose (closure of the) range is exactly equal to the vertical segment $[c- iT, c + iT]$. When $c=1$ (which corresponds to the pole of $\zeta (s)$ at $s=1$), instead of requiring that $\rho^{(T)}$ is continuous, we assume that $\rho^{(T)}$ has a meromorphic continuous to a connected neighborhood of $\{\Res =c \}$. Indeed, when $c \neq 1$ (resp., $c=1$), we can then apply the continuous (resp., meromorphic) version of the spectral mapping theorem (SMT) given in [HerLap1, App. E].\\

\tab Again, in light of the functional calculus form of the spectral theorem and of SMT, given in [HerLap1, App. E] (applied to the operator $\partial$ and to the function $\zeta = \zeta (s)$, which is continuous along $\sigma (\partial) = \{\Res = c \}$ if $c \neq 1$ and meromorphic in a connected open neighborhood of $\sigma (\partial)$ if $c = 1$, the pole of $\zeta$), we have (since $\mfa^{(T)} := \zeta (\partial^{(T)})$) that $\sigma (\mfa^{(T)}) = c \ell (\zeta (\sigma (\partial^{(T)})) = c \ell (\zeta ([c-iT, c+ iT]))$. Hence, we conclude (see \cite{HerLap1} for the proof) that the spectrum of $\mfa^{(T)}$ is given by $\sigma (\mfa^{(T)}) = \{\zeta(s) : \Res = c, |I_{m}(s)| \leq T \}$, which is a compact subset of $\mathbb{C}$ for $c \neq 1$. For $c=1$, the same formula holds provided one keeps the closure and also requires that $s \neq 1$ on the right side of the above expression for $\sigma (\mfa^{(T)})$. Alternatively, one can view $\zeta$ as a continuous $\widetilde{\mathbb{C}}$-valued function, with $\widetilde{\mbc} := \mbc \cup \{\infty \}$, and replace the left side by $\widetilde{\sigma} (\mfa^{(T)}) := \sigma (\mfa^{(T)}) \cup \{\infty \}$, the {\em extended spectrum} of $\mfa^{(T)}$. Note that for $c \neq 1, \mfa^{(T)}$ is bounded, so that $\widetilde{\sigma} (\mfa^{(T)}) := \sigma (\mfa^{(T)}),$ following the traditional notation.

\begin{remark}\label{R4.35}
Since, as follows easily from Thm. \ref{T4.1}, we have $\partial = c + iV$, where $V:= (\partial - c)/i$ is an unbounded self-adjoint operator with spectrum $\sigma (V)$ given by $\sigma (V) = \mbr$ $($and with no eigenvalues$)$, one can equivalently define $\partial^{(T)}$ by  $\partial^{(T)}: = c + i V^{(T)}$, where $V^{(T)} := \Phi^{(T)} (V)$ and $\Phi^{(T)}:\mbr \rightarrow \mbr$ satisfies $c \ell (\Phi^{(T)} (\mbr)) = [-T, T]$. In addition, when $c \neq 1$, one requires that $\Phi^{(T)}$ is continuous while when $c=1$, one assumes that $\Phi^{(T)}$ admits a meromorphic extension to a connected neighborhood of $\mbr$ in $\mbc$.
\end{remark}

\tab We can now introduce the following key definition. The spectral operator $\mfa$ is said to be {\em quasi-invertible} if each of its truncations $\mfa^{(T)}$ is invertible (for $T > 0$), in the usual sense (and hence, if for every $T > 0, 0 \notin \sigma (\mfa^{(T)})$. Using Thm. \ref{T4.4} and the above expression for $\sigma (\mfa^{(T)})$, along with the fact that an operator is invertible if its spectrum does not contain the origin, it is easy to check that if $\mfa$ is invertible, then it is quasi-invertible. The converse, however, need not be true, as we shall see in \S5. (It turns out that assuming that for all $c \in (0, 1/2),$ the quasi-invertibility of $\mfa= \mfa_c$ is equivalent to its invertibility, is equivalent to the Riemann hypothesis; see Thm. \ref{T5.4}.) \\

\tab The first result in this context is the exact operator theoretic analog of Thm. \ref{T3.4} (from \S 4$(c)$ above). 

\begin{theorem}[Riemann zeros and quasi-invertibility of $\mfa,$ \cite{HerLap1}, \cite{HerLap3}]\label{T4.6} Fix $c \in \mathbb{R}$. Then, the spectral operator $\mfa = \mfa_c$ is quasi-invertible if and only if $\zeta = \zeta(s)$ does not have any zeros on the vertical line $\{ \Res = c\}$. 
\end{theorem}

It follows that in the midfractal case when $c=1/2$, $\mfa_{1/2}$ is not quasi-invertible (since $\zeta$ has a zero along the critical line $\{\Res = 1/2\}$) and that for $c \geq 0$ (in order to avoid the trivial zeros of $\zeta$, located at $s= -2, -4, \cdots$), $\mfa$ is expected to be quasi-invertible everywhere else. (Recall that according to Hadamard's theorem (see, e.g., [Edw, Ti]) and the functional equation for $\zeta$, $\zeta = \zeta (s)$ does not have any zeros along the vertical lines $\{\Res =1 \}$ and $\{\Res =0 \}$.) In particular, we have the following operator theoretic counterpart of Thm. \ref{T3.5} (from \S 3$(b)$).

\begin{theorem}[Riemann hypothesis and quasi-invertibility of $\mfa$, \cite{HerLap1},\cite{HerLap3}]\label{T4.7} The spectral operator $\mfa = \mfa_c$ is quasi-invertible for all $c \in (0,1)$ $($other than for $c=1/2)$ $[$or equivalently, for all $c \in (0,1/2)$ or else, for all $c \in (1/2, 1)$, respectively$]$ if and only if the Riemann hypothesis is true.
\end{theorem}

Just as for Thm. \ref{T3.5}, we refer to Thm. \ref{T4.7} as a {\em symmetric criterion for the Riemann hypothesis} (the symmetry of the roles played by the intervals $(0,1/2)$ and $(1/2, 1)$ being due to the functional equation satisfied by $\zeta$). Finally, we close this section by noting that Thm. \ref{T4.7} also provides a precise and correct functional analytic formulation of the semi-heuristic result stated in Cor. 9.6 of [Lap-vFr2,3] about the ``invertibility'' of $\mfa$.

\section{Invertibility of the Spectral Operator and an Asymmetric Reformulation of the Riemann Hypothesis}

We have discussed in \S 3$(b)$ (based on [LapMai1,2]) and in \S 4$(d)$ (based on [HerLap1--3]) two different, but related (at least in spirit) {\em symmetric} reformulations of the Riemann hypothesis; see Thm. \ref{T3.5} (from \S 3$(b)$) and Thm. \ref{T4.7} (from \S 4$(d)$), respectively. Our goal in this section is to provide a new {\em asymmetric} reformulation of the Riemann hypothesis (due to the author), as well as to explain why it cannot be modified so as to become symmetric (with respect to the midfractal dimension $c=1/2$), without drastically changing the nature of the problem. As we shall see, this new reformulation is directly expressed in terms of the invertibility of $\mfa$. \\

\tab In \S 4$(d)$, we have studied the quasi-invertibility of the spectral operator $\mfa = \mfa_c$ (that is, the invertibility of each of its truncations $\mfa^{(T)} = \mfa_c^{(T)}$, where $T>0$ is arbitrary). In fact, both Thm. \ref{T4.6} and Thm. \ref{T4.7} (from \S 4$(d)$) are stated in terms of the quasi-invertibility of $\mfa_c$, either for a fixed $c \in (0,1)$, with $c \neq 1/2$ (in Thm. \ref{T4.6}) or (in Thm. \ref{T4.7}) for all $c \in (0,1)$ with $c \neq 1/2$ (equivalently, for all $c \in (0, 1/2)$ or else, for all $c \in (1/2, 1)$). (Recall from the discussion following Thm. \ref{T4.6} that $\mfa_{1/2}$ is {\em not} quasi-invertible and hence, not invertible.) However, although very convenient in the context of \S 4$(d)$, the new notion of ``quasi-invertibility'', because it involves a truncation of the infinitesimal shift $\partial$ (and hence, also of $\mfa$ as well as of its spectrum $\sigma (\mfa)$; see the beginning of \S 4$(d)$), is perhaps not so easy to grasp or to verify explicitly. Indeed, we lack concrete formulas for $\mfa^{(T)}$ whereas we have several explicit expressions for $\mfa$ to our disposal, even when $c \in (0,1)$. \\

\tab It is therefore natural to wonder whether the (usual) notion of invertibility of the spectral operator $\mfa = \mfa_c$ cannot be used to reformulate the Riemann hypothesis in this context. (See the discussion following Thm. \ref{T4.1} for the usual notion of invertibility of a possibly unbounded operator.) The answer to this question is affirmative, as we shall explain below; see Thm. \ref{T5.1}. Moreover, in light of Thm. \ref{T4.7} (based on the results of [HerLap1--4]), the asymmetry of the resulting criterion for RH turns out to be intimately connected to the universality of the Riemann zeta function $\zeta = \zeta (s)$ in the right critical strip $\{1/2 < \Res < 1 \}$ as well as to its (presumed) non-universality in the left critical strip $\{0 < \Res < 1/2 \}$. This last statement will become clearer as we progress in \S5. We can now state and prove the main result of this section. (See also Thm. \ref{T5.4}.)

\begin{theorem}[Asymmetric criterion for RH and invertibility of $\mfa$]\label{T5.1}
The spectral operator $\mfa = \mfa_c$ is invertible for all $c \in (0, 1/2)$ if and only the Riemann hypothesis is true. 
\end{theorem}

\tab By contrast, we have the following result, which was already observed in [HerLap1--4] and clearly shows that the reformulation of RH obtained in Thm. \ref{T5.1} is asymmetric, in a strong sense.

\begin{theorem}\label{T5.2}
The spectral operator $\mfa = \mfa_c$ is not invertible for {\em any} $c \in (1/2, 1)$.
\end{theorem}

\begin{proof}[Proof of Thms. \ref{T5.1} and \ref{T5.2}]
(i) Let us first prove Thm. \ref{T5.1}. Recall that $\mfa$ is invertible if and only if $0 \notin \sigma(\mfa)$. In light of Thm. \ref{T4.4}, we know that for any $c \in \mathbb{R}$ (with $c \neq 1$) and with $\mfa = \mfa_c$, as usual, $\sigma (\mfa_c)$ coincides with the closure of the range of $\zeta = \zeta(s)$ along the vertical line $\{ \Res = c \}$. Now, according to a result of Garunk\v stis and Steuding in \cite{GarSte} concerning the non-universality of $\zeta$ in the left critical strip $\{ 0 < \Res < 1/2 \}$, we know that conditionally (i.e., under the Riemann hypothesis), we have that for all $c \in (0, 1/2), \sigma (\mfa)$ (that is, the closure of the range of $\zeta$ on $\{\Res = c \}$) is a strict subset of $\mathbb{C}$ and, in fact, that $0 \notin \sigma(\mfa)$; see [GarSte, Lemma 4 \& Prop. 5] and their proofs. Hence, $\mfa$ is invertible.\\

\tab Conversely, assume that $\mfa = \mfa_c$ is invertible (i.e., $0 \notin \sigma (\mfa)$) for every $c \in (0, 1/2)$. Then, since (by Thm. \ref{T4.4}) $\sigma (\mfa) \supseteq \{\zeta (s): \Res = c \}$, it follows that $\zeta(s) \neq 0$ for all $s \in \mathbb{C}$ with $\Res = c$ and every $c \in (0, 1/2)$. In light of the functional equation for $\zeta$, we then deduce that $\zeta (s) \neq 0$ for all $s \in \mathbb{C}$ with $0 < \Res < 1, \Res \neq 1/2$; i.e., the Riemann hypothesis holds. This concludes the proof of Thm. \ref{T5.1}. \\

(ii) According to the Bohr--Courant theorem \cite{BohCou}, which itself is implied by the universality of $\zeta$ in the right critical strip $\{1/2 < \Res <1 \}$ (see, e.g., \cite{Ste} for an exposition), the range of $\zeta$ is dense along every vertical line $\{ \Res = c \}$, with $1/2 < c < 1$; i.e., in light of Thm. \ref{T4.4}, $\sigma (\mfa) = \mathbb{C}$ and hence, $0 \in \sigma (\mfa)$. Therefore, the conclusion of Thm. \ref{T5.2} holds: $\mfa$ is not invertible for any $c \in (0, 1/2)$, as desired. This completes the proof of Thm. \ref{T5.2}.
\end{proof}

\begin{remark}\label{R5.3}
$($a$)$ The universality of $\zeta$ implies a much stronger result than the one used in the proof of Thm. \ref{T5.2} $($part $($ii$)$ of the above proof$)$. Namely, for every integer $n \geq 0$  and every $c \in (1/2, 1)$, the range  $($along the vertical line $\{\Res =c \})$ of $(\zeta(s), \zeta'(s), \cdots, \zeta^{(n)} (s))$ is dense in $\mathbb{C}^{n+1}$, where $\zeta^{(k)} (s)$ denotes the $k$-th complex derivative of $\zeta$ at $s \in \mathbb{C}$ and $\zeta^{(0)} := \zeta$. This generalization of the Bohr--Courant theorem $($from $[$BohCou$])$ is due to Voronin in $[$Vor1$]$ and was probably a key motivation for Voronin's discovery in $[$Vor2$]$ of the universality of $\zeta$ $($in the right critical strip$).$ The universality theorem was extended to $($suitable$)$ compact subsets of the right critical strip by Bagchi and Reich in $[$Bag$]$ and $[$Rei$].$ We refer, e.g., to the books $[$KarVor, Lau, Ste$]$ along with $[$HerLap1,4$]$ for many other relevant references and extensions of the universality theorem to other $L$-functions. \\

\tab $($b$)$ The universality of $\zeta$ in the right critical strip roughly means that given any compact subset $K$ of $\{1/2 < \Res < 1 \}$ with connected complement in $\mathbb{C}$, and given any nowhere vanishing $($complex-valued$)$ continuous function $g$ on $K$ that is holomorphic on the interior of $K$ $($which may be empty$),$ then $g$ can be uniformly approximated by some $($and, in fact, by infinitely many$)$ vertical translates of $\zeta$ on $K$. \\

\tab $($c$)$ A quantized $($i.e., operator-valued$)$ version of universality is provided in $[$HerLap1,4$].$ Interestingly, in that context, the natural replacement for the complex variable $s$ is the family of truncated spectral operators $\{\partial^{(T)} \}_{T>0}$.
\end{remark}

\tab The spectral operator $\mfa = \mfa_c$ is bounded (and hence invertible, with a compact spectrum $\sigma (\mfa)$ not containing the origin), for every $c > 1$; it is unbounded (and therefore its spectrum $\sigma (\mfa)$ is a closed unbounded subset of $\mathbb{C}$) for every $c \leq 1$. It is not invertible for $1/2 < c <1$ since $0 \in \sigma (\mfa) = \mathbb{C}$; by contrast, according to Thm. \ref{T5.1}, it is invertible (i.e., $0 \notin \sigma (a)$) for every $c \in (0, 1/2)$ if and only  if the Riemann hypothesis is true. For a discussion of the mathematical phase transition occurring (conditionally) in the midfractal case when $c = 1/2$, we refer to [Lap2,3] and [HerLap1,2]. In the present context of \S5, we also refer to the brief discussion at the end of the introduction (\S1). \\

\tab Let $\mathfrak{b} = \mfa^* \mfa \ (= \mfa \mfa^*$, since $\mfa$ is normal). Then $\mathfrak{b} = \mathfrak{b}_c$ is a nonnegative self-adjoint operator. It is bounded (resp., invertible) if and only if $\mfa$ is; that is, iff $c \geq 1$ (resp., $c > 1$ or, conditionally, $0 < c < 1/2$). Furthermore, $\mathfrak{b}$ is invertible if and only if it is bounded away from zero (i.e., $\mathfrak{b} \geq \alpha$, for some constant $\alpha > 0$, which may depend on $c$). We may now close this section by stating the following theorem, which provides several reformulations of the Riemann hypothesis discussed in this paper, as well as a new one (the last one given in Thm. \ref{T5.4}).

\begin{theorem}\label{T5.4}
The following statements are equivalent$:$
\begin{enumerate} \itemsep1pt \topsep0pt   
\item[$($i$)$] The Riemann hypothesis is true.
\item[$($ii$)$] The spectral operator $\mfa$ is invertible for every $c \in (0, 1/2)$.
\item[$($iii$)$] The spectral operator $\mfa$ is quasi-invertible for every $c \in (0, 1/2)$ $($or equivalently, for every $c \in (1/2, 1))$.
\item[$($iv$)$] For every $c \in (0, 1/2)$, the unbounded, nonnegative, self-adjoint operator $\mathfrak{b}$ is bounded away from zero $($i.e., is invertible$).$
\end{enumerate}
\end{theorem}

\begin{proof}
(i) $\Leftrightarrow$ (ii): This follows from Thm. \ref{T5.1}. \\
(i) $\Leftrightarrow$ (iii): This follows from Thm. \ref{T4.7}. \\
(ii) $\Leftrightarrow$ (iv): This follows from the discussion immediately preceding the statement of this theorem.
\end{proof}

We note that, as was discussed earlier, the equivalence between (i) and (iii) is a symmetric criterion for RH. By contrast, the equivalence between (i) and (ii) is an asymmetric criterion for RH. Therefore, so is the equivalence between (i) and (iv), which, in practice, might provide the most hopeful way to try to prove the Riemann hypothesis along those lines. The author has obtained some preliminary (but not definitive) results in this direction, based on explicit computations and by using a suitable class of ``test functions'' in $D(\mathfrak{b})$, the domain of $\mathfrak{b}$, but it is premature to discuss them here.

\begin{remark}\label{R5.4}
$($a$)$ It is natural to wonder what is the inverse of $\mfa = \mfa_c$ when $0 < c <1/2.$ We conjecture that, under RH, it is given by $\mfa^{-1}(f) = \sum_{n=1}^\infty \mu (n) n^{-\partial} (f)$ $($so that $\mfa^{-1} (f) (t) = \sum_{n=1}^\infty \mu (n) f (t - \log n))$, for all $f \in D(\mfa) = \mathbb{H}_c$, where $\mu = \mu(n)$ is the M\" obius function $($compare with part $($3$)$ of \S 4$(c))$. Accordingly, the quantized Riemann zeta function $\mfa = \zeta (\partial)$ would behave very differently from the ordinary $($complex-valued$)$ Riemann zeta function $\zeta = \zeta (s)$. Indeed, even under RH, the Dirichlet series $\sum_{n=1}^\infty \mu (n) n^{-s}$ cannot converge for any $s_0 \in \mathbb{C}$ such that $0 < Re (s_0) < 1/2$. If it did, then its sum would have to be holomorphic $($see $[$HardWr$]$ or $[$Ser, \S VI.2$])$ and to coincide with $(1/\zeta)(s)$ on $\{\Res > Re (s_0) \}$, which is impossible because $1/\zeta$ must have a pole at every zero of $\zeta$ along the critical line $\{\Res = 1/2 \}.$ $($See also, e.g., $[$Edw$].)$ \\

\tab $($b$)$ The same methods as those discussed in this paper can be applied to a very large class of $L$-functions $($or arithmetic zeta functions$),$ all of which are expected to satisfy the generalized Riemann hypothesis $($GRH$);$ see, e.g., $[$Sarn, ParSh$],$ $[$Lap-vFr3, App. A$]$ and $[$Lap6, Apps. B,C \& E$].$ For obtaining the analog of the results of \S 5, we would need to use, in particular, the universality results discussed in $[$Ste$].$ The relevant results of $[$GarSte$]$ used in the proof of Thm. \ref{T5.1} would also need to be extended to this broader setting.\\

\tab $($c$)$ In $[$Lap8,9$]$, the author has adapted $($and extended in various directions$),$ the present theory to a new functional analytic framework $($based on weighted Bergman spaces of entire functions, see $[$HedKorZh, AtzBri$]).$ At this early stage, it seems that depending on the goal being pursued, one approach or the other may be more appropriate. In particular, the approach discussed in the present article $($beginning with \S4$)$ is ideally suited for formulating and obtaining the new results of \S5.
\end{remark}


\section*{Acknowledgment}

This research was partially supported by the US National Science Foundation (NSF) under grants DMS-0707524 and DMS-1107750 (as well as by many earlier NSF grants since the mid-1980s). Part of this work was completed while the author was a Visiting Professor at the Institut des Hautes Etudes Scientifiques (IHES) in Paris/Bures-sur-Yvette, France.

\section{Glossary}

\begin{itemize}
\setlength\itemsep{.1em}
\item[] abscissa of convergence\dotfill 5
\item[] $a_k \sim b_k$ as $k \rightarrow \infty$, asymptotically equivalent sequences\dotfill 10 
\item[] $c \ell (A)$, closure of a subset $A$ in $\mbc$ (or in $\mbr^N$)\dotfill 15
\item[] $\mfa = \mfa_c$, spectral operator\dotfill 5,13,15
\item[] $\mfa^{(T)} = \mfa_c^{(T)}$, truncated spectral operator\dotfill 12
\item[] $\mfb = \mfa \mfa^* = \mfa^* \mfa$, where $\mfa$ is the spectral operator\dotfill 5,19
\item[] complex dimensions (of a fractal string)\dotfill 7
\item[] CS, Cantor string\dotfill 7
\item[] $\mbc$, the field of complex numbers\dotfill 5
\item[] $\widetilde{\mbc} := \mbc \cup \{\infty \}$, the Riemann sphere\dotfill 16
\item[] Dirichlet series\dotfill 7,15,19
\item[] $D$, Minkowski dimension of a fractal string (or of $A \subset \mbr^N$)\dotfill 5,11
\item[] $\mcd= \mcd_\mcl$, the set of complex dimensions of a fractal string $\mcl$ (or of its boundary)\dotfill 7 
\item[] $dist (x, A) = \inf \{|x-a|: a \in A \}$, Euclidean distance from $x$ to $A \subset \mbr^N$\dotfill 5
\item[] $\partial = \partial_c$, infinitesimal shift of the real line\dotfill 14
\item[] $\partial^{(T)} = \partial_c^{(T)}$, truncated infinitesimal shift\dotfill 16
\item[] $\partial \Omega$, boundary of a subset $\Omega$ of $\mbr^N$ (or of a fractal string)\dotfill 5,11
\item[] $\delta_x$, Dirac measure at the point $x$\dotfill 8
\item[] Euler product\dotfill 7,15
\item[] $|E| = |E|_N, N$-dimensional Lebesgue measure of $E \subset \mbr^N$\dotfill 5,11
\item[] $\Gamma (t) := \int_0^{+ \infty} x^{t-1} e^{-x} dx$, the gamma function\dotfill 7
\item[] $\mbh_c := L^2 (\mbr, e^{-2ct}dt)$, weighted Hilbert space\dotfill 14
\item[] $i = \sqrt{-1}$, imaginary unit\dotfill 7
\item[] (ISP)$_D$, inverse spectral problem for the fractal strings of Minkowski dimension $D$\dotfill 11  
\item[] $\xi = \xi (s)$, completed (or global) Riemann zeta function\dotfill 7
\item[] $\log_a x$, the logarithm of $x > 0$ with base $a > 0$; $y = \log_a x \Leftrightarrow x = a^y$\dotfill 7
\item[] $\log x := \log_e x$, the natural logarithm of $x; \, y = \log x \Leftrightarrow x = e^y$\dotfill 8
\item[] $\mcl = \{\ell_j \}_{j=1}^\infty$, a fractal string with lengths $\ell_j$\dotfill 5
\item[] Minkowski content (of a Minkowski measurable fractal string)\dotfill 6
\item[] Minkowski nondegenerate (fractal string)\dotfill 6
\item[] Minkowski measurable (fractal string)\dotfill 6
\item[] $\mcm_{*,d}$ and $\mcm_d^*$, lower and upper $d$-dimensional contents of a fractal string $\mcl$ (or of a bounded set $A \subset \mbr^N$) $\ldots$\dotfill 5,11
\item[] $\mcm, \mcm_*$ and $\mcm^*$, Minkowski content, lower and upper Minkowski content (of $\mcl$ or of $A \subset \mbr^N$)\dotfill 6,11
\item[] $\mu (n)$, M\" obius function\dotfill 19
\item[] $N_\mcl$, geometric counting function of a fractal string $\mcl$\dotfill 5--6
\item[] $N_\nu$, spectral (or frequency) counting function of a fractal string or drum\dotfill 7,11
\item[] $\Omega_\varepsilon$, $\varepsilon$-neighborhood of a fractal string (or of a fractal drum)\dotfill 5,11 
\item[] $\mathcal{P}$, the set of prime numbers\dotfill 7
\item[] $res (f; \omega)$, residue of the meromorphic function $f$ at $\omega \in \mbc$\dotfill 8
\item[] $\sigma(\mcl)$, spectrum of a fractal string $\mcl$\dotfill 6
\item[] $\sigma (T)$, spectrum of the operator $T$\dotfill 14
\item[] $\widetilde{\sigma} (T)$, extended spectrum of the operator $T$\dotfill 16
\item[] $V (\varepsilon)$, volume of the $\varepsilon$-neighborhood of a fractal string (or drum)\dotfill 5,11
\item[] $W (x)$, Weyl term\dotfill 10--11
\item[] $[x]$, integer part of $x \in \mbr$\dotfill 10
\item[] $\zeta (s) = \sum_{j=1}^\infty j^{-s}$, Riemann zeta function\dotfill 6--7 
\item[] $\zeta_\mcl (s) = \sum_{j=1}^\infty (\ell_j)^s$, geometric zeta function of a fractal string $\mcl = \{\ell_j \}_{j=1}^\infty$\dotfill 6 
\item[] $\zeta_\nu = \zeta_{\nu, \mcl}$, spectral zeta function of a fractal string $\mcl$\dotfill 6,10 
\item[] $||\cdot||_c$, norm of the Hilbert space $\mbh_c$\dotfill 14
\item[] $<\cdot,\cdot>_c$, inner product of the Hilbert space $\mbh_c$\dotfill 14
\end{itemize}

\end{document}